\documentclass[10pt,journal,compsoc]{IEEEtran}
\usepackage{tabularx}
\usepackage{graphicx}
\usepackage{url}
\usepackage{subcaption}
\usepackage{hyperref}
\usepackage{amssymb}
\usepackage{amsthm}
\usepackage{array}
\usepackage{mdwlist}
\usepackage{amsmath}
\usepackage{float}
\hypersetup{colorlinks=false,linkcolor=blue,urlcolor=blue,citecolor=red}
\usepackage[labelfont=bf, textfont=bf]{caption}
\usepackage[linesnumbered,ruled]{algorithm2e}
\usepackage{algorithmic}
\usepackage{enumitem}
\setlist{leftmargin=5mm}
\usepackage{epstopdf}

\newcolumntype{C}[1]{>{\centering\let\newline\\\arraybackslash\hspace{0pt}}m{#1}}
\newcommand{\ie}{\emph{i.e., }}
\newcommand{\eg}{\emph{e.g., }}
\newcommand{\etal}{\emph{et al.}}
\newcommand{\st}{\emph{s.t. }}
\newcommand{\wrt}{\emph{w.r.t. }}
\newcommand{\cf}{\emph{cf. }}

\DeclareMathOperator*{\argmax}{arg\,max}
\DeclareMathOperator*{\pp}{\textbf{p}} 
\DeclareMathOperator*{\uu}{\textbf{u}} 
\DeclareMathOperator*{\p0}{\textbf{p}^0} 
\DeclareMathOperator*{\u0}{\textbf{u}^0} 

\begin{document}
\title{BiRank: Towards Ranking on Bipartite Graphs}
\author{Xiangnan~He, Ming~Gao~\IEEEmembership{Member,~IEEE}, Min-Yen~Kan~\IEEEmembership{Member,~IEEE} and Dingxian Wang
\IEEEcompsocitemizethanks{
	\IEEEcompsocthanksitem X. He is with Web IR/NLP group, School of Computing, National University of Singapore, Singapore. E-Mail: xiangnanhe@gmail.com 
	\IEEEcompsocthanksitem M. Gao is with Software Engineering Institute, East China Normal University, China. E-Mail: mgao@sei.ecnu.edu.cn 
	\IEEEcompsocthanksitem M-Y. Kan is with Web IR/NLP group, School of Computing, National University of Singapore, Singapore. E-Mail: kanmy@comp.nus.edu.sg
	\IEEEcompsocthanksitem D. Wang is with Ranking team, Search Science Department, eBay Inc, China. E-Mail: diwang@ebay.com
}
}

\markboth{IEEE Transactions on Knowledge and Data Engineering, Submission 2016}
{}

\IEEEtitleabstractindextext{
\begin{abstract}
	The bipartite graph is a ubiquitous data structure that can model the relationship between two entity types: for instance, users and items, queries and webpages.
	In this paper, we study the problem of ranking vertices of a bipartite graph, based on the graph's link structure as well as prior information about vertices (which we term a \textit{query vector}).
	We present a new solution, BiRank, which iteratively assigns scores to vertices and finally converges to a unique stationary ranking.
	In contrast to the traditional random walk-based methods, BiRank iterates towards optimizing a regularization function, which smooths the graph under the guidance of the query vector. Importantly, we establish how BiRank relates to the Bayesian methodology, enabling the future extension in a probabilistic way. 
	To show the rationale and extendability of the ranking methodology, we further extend it to rank for the more generic $n$-partite graphs. 
	BiRank's generic modeling of both the graph structure and vertex features enables it to model various ranking hypotheses flexibly. 
	To illustrate its functionality, we apply the BiRank and TriRank (ranking for tripartite graphs) algorithms to two real-world applications: a general ranking scenario that predicts the future popularity of items, and a personalized ranking scenario that recommends items of interest to users. 
	Extensive experiments on both synthetic and real-world datasets demonstrate BiRank's soundness~(fast convergence), efficiency~(linear in the number of graph edges) and effectiveness~(achieving state-of-the-art in the two real-world tasks).
\end{abstract}

 }

\maketitle

\IEEEdisplaynotcompsoctitleabstractindextext
\IEEEpeerreviewmaketitle

\section{Introduction}
\label{sec:intro}
\IEEEPARstart
Graphs provide a universal language to represent relationships between entities. In real-world applications, not only should the relationships between entities of the same type be considered,
but the relationships between different types of entities should also be modeled. Such relationships naturally form a bipartite graph, containing rich information to be mined from.  For example, in YouTube, the videos and users form a bipartite relationship where edges indicate a viewing action; in Web search, the relationships between queries and search engine result pages are user actions (``clicks''), which provide important relevance judgments from the user's perspective.

A fundamental task in the mining of bipartite graphs is to rank vertices against a specific criterion. 
Depending on the setting, assigning each vertex a ranking score can be used for many tasks, including the estimation of vertex importance (popularity prediction) and the inference of similar vertices to a target vertex (similarity search), and edge suggestion for connecting a target vertex (link prediction and recommendation).
Existing work on graph ranking have largely focused on unipartite graphs, including PageRank~\cite{pagerank:1999}, HITS~\cite{HITS:1999}\footnote{Note that although HITS does handle bipartite graphs, the algorithm was designed for ranking on unipartite graphs by treating vertices with two roles -- hub and authority.}, and many of their variants~\cite{Haveliwala:2002,SALSA:2000, liu2008browserank,gao2011semi-pagerank}. Although several works~\cite{Co-HITS:2009, Sun:2005, li2014spam} have considered ranking on bipartite graphs, they have either focused on a specific application or adapted existing algorithms to handle the bipartite case.  In our opinion, the work up to the current time, lacks a thorough theoretical analysis. 

In this paper, we focus on the problem of ranking vertices of bipartite graphs. We formulate the ranking problem in a generic manner -- accounting for both the graph's structural information and the proper incorporation of any prior information for vertices, where such vertex priors can be used to encode any features of vertices. The main contributions of this paper are summarized as follows:
\begin{itemize}
	\item We develop a new algorithm -- BiRank -- for addressing the ranking problem on bipartite graphs, and show its convergence to a unique stationary point;
	\item We analyze BiRank through the formalism of graph regularization, and present a complementary Bayesian view. These two views enable future extensions to be grounded and compelling from a theoretically principled way (either algebraically or probabilistically).
	\item We deploy BiRank to the general ranking scenario of item popularity prediction, illustrating how to parameterize it to encode several ranking hypotheses;
	\item We extend the methodology to rank on the more generic $n$-partite graphs, and employ it for a personalized ranking scenario by mining tripartite graphs.
	\item We conduct extensive experiments to justify our methods for the two real-world ranking scenarios of popularity prediction and personalized recommendation.
\end{itemize}

The paper is organized as follows.
After reviewing related works in Section~\ref{sec:related}, we formalize the problem in Section~\ref{sec:pre}.
Then we describe the BiRank algorithm in Section~\ref{sec:algorithm}, and interpret it from two views in Section~\ref{sec:regu}. 
In Section~\ref{sec:applications}, we discuss how to apply BiRank to popularity prediction and personalized recommendation. We conduct experiments in Section~\ref{sec:exper},
before concluding the paper in Section~\ref{sec:conclusion}. 

\section{Related Work}
\label{sec:related}
BiRank, which ranks vertices of a bipartite graph, can be categorized as a link-based object ranking method under the paradigm of link mining~\cite{Survey:link_mining2005}. 
In this section, we focus on related work that contribute in the ranking method, and omit discussion of other relevant issues such as efficiency and evolving graphs.
We then review work that can benefit from such bipartite graph ranking, forming the potential downstream applications of BiRank. 

\subsection{Graph Ranking Methods}
In the context of web graph ranking, PageRank~\cite{pagerank:1999} and HITS~\cite{HITS:1999} are the most prominent methods. PageRank estimates the importance score of vertices as the stationary distribution of a random walk process -- starting from a vertex, the surfer randomly jumps to a neighbor vertex according to the edge weight. 
HITS assumes each vertex has two roles: \textit{hub} and \textit{authority}, transforming the graph to a bipartite graph. A vertex has a high authority score if it is linked by many vertices with hub score, and a vertex has a high hub score if it links to many authoritative vertices. 

Many variants start from the basic themes of PageRank and HITS. Ng \etal~\cite{ng2001stable} studied the stability of the two algorithms, finding HITS more sensitive to small perturbations in the graph structure under certain situations. They proposed two variants --- Randomized HITS and Subspace HITS --- that yield more stable rankings. Similarly, Lempel \etal~\cite{SALSA:2000} found that applying HITS on graphs with TKCs~(\textit{tightly knit communities}, \ie small but highly interconnected set of vertices) fails to identify meaningful authority vertices. They devised SALSA as a stochastic variant of HITS, for alleviating the TKC effect. Haveliwala~\cite{Haveliwala:2002} proposed topic-sensitive PageRank~(also known as \textit{personalized PageRank}) by replacing the uniform teleportation vector with a non-uniform vector that encodes each vertex's topic score~(\cf query vector in our BiRank context). 
Later on, Ding \etal~\cite{Ding:2002} unified HITS and PageRank under a normalized ranking framework.
Inspired by the discrete-time Markov process explanation of PageRank, Liu \etal~\cite{liu2008browserank} also proposed BrowseRank based on continuous time Markov process, exploiting user behavior data for page importance ranking. To incorporate side information on vertices and edges into ranking, Gao \etal~\cite{gao2011semi-pagerank} extended PageRank in a semi-supervised way by learning the transition matrix based on the features on vertices and edges. 

Along a separate line of work -- ranking on graphs based on regularization theory~\cite{Zhou04learningwith, Smola:2003, Agarwal:2006} -- has gained popularity within the machine learning community. These works mainly consider the problem of labeling vertices of a graph from partially known labels, also termed \textit{semi-supervised learning} or \textit{manifold learning} on graphs. Smola \etal~\cite{Smola:2003} summarized early works on graph kernels~(\eg Diffusion kernels), and formulated a family of regularization operators on graphs to encompass such kernels. 
Inspired by it, Zhou \etal~\cite{Zhou04learningwith} developed a regularization framework consisting of two constraints: smoothness and fitting,
proposing an iterative algorithm \cite{Zhou04rankingon} for optimizing the regularization function. 
Later on, they \etal~supplemented the regularization framework by developing a discrete analytic theory of graphs~\cite{zhou:2005} and extending it to cover directed graphs~\cite{Zhou:2005:directed1}. Agarwal~\cite{Agarwal:2006} further extended the regularization framework by replacing the fitting term~(\ie sum of squared errors) to the \textit{hinge} ranking loss, proposing an algorithm with similarities to solving \textit{support vector machine} to optimize the regularization function. 

The above discussed works have all focused on ranking for homogeneous graphs, where vertices are of the same type. Our proposed BiRank targets the task of ranking for bipartite graphs, where vertices are of two different types.
Separately handling the two vertex types is very important for ranking in bipartite graphs for many applications, a claim we validate through our experiments later.
Inspired by the graph regularization framework \cite{zhou:2005}, we develop the BiRank algorithm, which can be seen an extension of the manifold ranking algorithm \cite{Zhou04rankingon} on bipartite graphs. 

\subsection{Ranking on Bipartite Graphs}
There are other algorithms developed for bipartite graph ranking that target specific applications.
As a natural way to represent relationship between two types of entities, bipartite graphs have been widely used across domains. As a consequence, ranking on bipartite graph data have been explored to address many applications. 
For example, in Web search, Deng \etal~\cite{Co-HITS:2009} modeled queries and URLs for query suggestion, Cao \etal~\cite{Cao2011} considered the co-occurrence between entities and queries for entity ranking, Li \etal~\cite{li2014spam} modeled users and their search sessions for detecting click spam, and Rui \etal~\cite{rui2007bipartite} mined visual features and the surrounding texts for Web image annotation. In practical recommender systems, bipartite graphs methods have been used for Twitter user recommendation~\cite{Gupta:2013} and YouTube video recommendation~\cite{Baluja:2008}. In the domain of natural language processing, Parveen \etal~\cite{parveen2014multi} generated multi-document summarization based on the relationship of sentences and lexical entities.

In terms of the ranking technique, these works share the same cornerstone --- they all rank by iteratively propagating scores on the graph; either through a PageRank-like random walk or a HITS-like iterative process -- which is adjusted for use on bipartite graphs. The prominent advantage of such propagation-based methods is that the global structure of the graph can be implicitly considered, which is an effective way to deal with the data sparsity and make use of the graph structure. Similar to their ranking algorithms, our proposed BiRank is also a propagation-based method; however, the main difference lies in the normalization strategy used in the iterative process. The symmetric normalization used in BiRank normalizes an edge weight by both of its vertex ends, which accords a smoothing on the graph that can be explained by the regularization theory~\cite{zhou:2005}. As a result, extensions to the algorithm, such as incorporating more features about vertices and edges, can be achieved in a theoretically principled way.
More importantly, we believe such a bridge with the graph regularization theory and Bayesian framework allows BiRank a broader algorithmic extensions that are difficult to achieve by PageRank, HITS and their variants. For example, we can adjust the score propagation process to use a different ranking-based objective~(Section~\ref{ss:interpret}), and learn the combination parameters in an automatic way~(Section~\ref{ss:bayesian}). We will study these algorithmic extensions of BiRank in the future.





\section{Problem Formulation}
\label{sec:pre}
We first present the bipartite graph model and then give the notation convention used. We then formalize the ranking problem that we address in this paper. 

\noindent\textbf{Notations}. Let $G = (U\cup P, E)$ be a bipartite graph, where $U$ and $P$ represent the vertex sets of the two entity types respectively, and $E$ represents the edge set ({\it n.b.}, bipartite graphs have edges only between vertices of the two different types).
Figure~\ref{fig:bipartite} shows an example of the bipartite structure.
\vspace{-15pt}
\begin{figure}[h]
	\centering
	\includegraphics[width=0.3\textwidth]{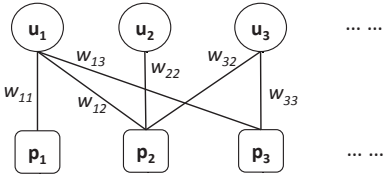}
	\caption{Bipartite User--Item Structure.}
	\label{fig:bipartite}
\end{figure}
\vspace{-5pt}

We use $u_i$ to denote the $i$-th vertex in $U$, and $p_j$ to denote the $j$-th vertex in $P$, where $1\le i\le |U|$ and $1\le j\le |P|$; set cardinality $|U|$ denotes number of elements in $U$. Edges carry 
non-negative weights $w_{ij}$, modeling the relationship strength between the connected vertices $u_i$ and $p_j$ (if $u_i$ and $p_j$ are not connected, their edge weight $w_{ij}$ is zero). 
As such, we can represent all edge weights of the graph as a $|U|\times|P|$ matrix $W = [w_{ij}]$. 
For each vertex $u_i$, we denote its weighted degree (\ie sum of connected edges' weights) as $d_i$, and use a diagonal matrix $D_u$ to denote the weighted degrees of all vertices in $U$ such that $(D_u)_{ii} = d_i$; and similarly, for $d_j$ and $D_p$. 
Note that in this paper, we deal with undirected bipartite graphs, \ie we do not model any directionality in the edges. 


\noindent\textbf{Problem Definition.}
In a nutshell, the general graph ranking problem is to assign each vertex a score \st a given expectation is satisfied. For example, PageRank~\cite{pagerank:1999} infers an importance score for each vertex to capture the intuition that an important vertex should be linked by many other important vertices. 
As in many applications, a ranking simply based on the graph structure is insufficient; often, there also exists some prior information (or features) on the vertices. For example, in webpage ranking, we already know some webpages are important~(\eg official sites), and wish to incorporate this knowledge into the ranking process;
in the application of recommendation, we need to consider a user's historical actions as the prior knowledge of the user's preference. 
We term such prior knowledge as a \textit{query vector}, which encodes the prior belief of the score of vertices with respect to the ranking criterion. 
In this paper, we study the bipartite graph ranking problem where a query vector is given, formally defined as:
\begin{description}
	\item[\textbf{Input}:] A bipartite graph $G = (U\cup P, E)$ with its weight matrix $W$.  A query vector $\u0,\p0$ encodes the prior belief concerning the vertices in $U$ and $P$, respectively, with respect to the ranking criterion.
	\item[\textbf{Output}:] A function $f : P\cup U \to \mathbb{R}$, which maps each vertex in $G$ to a real number.
\end{description}
\noindent The function value $f(u_i)$ and $f(p_j)$ form the ranking score of vertex $u_i$ and $p_j$, respectively. To keep the notation simple, we also use $u_i$ and $p_j$ to denote the ranking score, and represent the final ranking score of all vertices as two ranking vectors $\uu=[u_i]$ and $\pp=[p_j]$.

\section{Iterative BiRank}
\label{sec:algorithm}
With the preliminaries settled, we now detail the iterative paradigm of the BiRank algorithm.
We first describe how we design the ranking algorithm, analyzing its time complexity. Then we study its convergence properties in theory. Finally we discuss the connection of BiRank with other similarly-styled iterative bipartite graph ranking algorithms. 

\subsection{BiRank's Design}

To rank vertices based on the graph structure, seminal algorithms like PageRank and HITS have been proposed. Motivated from their design, our intuition for bipartite graph ranking is that the scores of vertices should follow a smoothness convention, namely that: \textit{a vertex~(from one side) should be ranked high if it is connected to higher-ranked vertices (from the other side)}. This rule defines a mutually-reinforcing relationship, which is naturally implemented as an iterative process that refines each vertex's score as the sum of the contribution from its connected vertices:
\begin{equation*}
p_j = \sum_{i=1}^{|U|} w_{ij} u_i; \quad u_i = \sum_{j=1}^{|P|} w_{ij} p_j.
\end{equation*}
As it is an additive update rule, normalization is necessary to ensure the convergence and stability. Two strategies have been widely adopted in previous work: 1) a PageRank-style that normalizes $W$ (and $W^T$) to a stochastic matrix, leading to a probabilistic random walk explanation; and 2) a HITS-style method that normalizes the ranking scores of vertices after each iteration.
In our BiRank method, we adopt the symmetric normalization scheme, which is inspired from Zhou {\it et al.}'s work \cite{Zhou04learningwith} addressing  semi-supervised learning on graphs. The idea is to smooth an edge weight by the degree of its two connected vertices simultaneously:
\begin{equation}
\label{eq:simple1}
p_j = \sum_{i=1}^{|U|} \frac{w_{ij}}{\sqrt{d_i} \sqrt{d_j}} u_i; \quad u_i = \sum_{j=1}^{|P|} \frac{w_{ij}}{\sqrt{d_i} \sqrt{d_j}} p_j,
\end{equation}
\noindent where $d_i$ and $d_j$ are the weighted degrees of vertices $u_i$ and $p_j$, respectively. 
The use of symmetric normalization is a key characteristic of BiRank, allowing edges connected to a high-degree vertex to be suppressed through normalization, lessening the contribution of high-degree vertices. This has the beneficial effect of toning down the dependence of top rankings on high-degree vertices, a known defect of the random walk-based diffusion methods~\cite{Baluja:2008}.  This gives rise to better quality results. 



To account for the query vector $\p0$ and $\u0$ that encode the prior belief on the importance of the vertices, one can either opt for 1) incorporating the graph ranking results for combination in post-processing ({\it a.k.a} late fusion), or 2) factoring the query vector directly into the ranking process. The first way of post-processing yields a ranking that is a compromise between two rankings; for scenarios that the query vector defines a full ranking of vertices, this ensemble approach might be suitable. However, when the query vector only provides \textbf{partial} information -- {\it i.e.}, only a small proportion of vertices have a prior score while most other vertices have no prior information -- this method fails to identify an optimal ranking. 
For example, in the application of personalized recommendation (see Section~\ref{ss:recom}), the aim is to rank \textit{unconsumed items} for a user; the query vector encodes the user's known preference, which is a sparse vector with the consumed items as non-zeros. 
In this case, simply combining the ranking from graph structure and query vector via post processing does not work, since the ranking of \textit{unconsumed items} will solely depend on the graph structure. 
As such, in BiRank we opt for the second way that factors the query vector directly into the ranking process, which has the advantage of using the query vector to guide the ranking process:
\begin{equation}
\begin{aligned}
\label{eq:birank1}
p_j &= \alpha\sum_{i=1}^{|U|} \frac{w_{ij}}{\sqrt{d_i} \sqrt{d_j}} u_i + (1-\alpha)p_j^0; \\
u_i &= \beta\sum_{j=1}^{|P|} \frac{w_{ij}}{\sqrt{d_i} \sqrt{d_j}} p_j + (1-\beta)u_i^0,
\end{aligned}
\end{equation}
where $\alpha$ and $\beta$ are hyper-parameters to weight the importance of the graph structure and the prior query vector, to be set between $[0,1]$. To keep notation simple, we can also express the iteration in its equivalent matrix form:
\begin{equation}
\begin{aligned}
\label{eq:birank2}
\textbf{p} &= \alpha S^T \textbf{u} + (1-\alpha)\p0; \\
\textbf{u} &= \beta S \textbf{p} + (1-\beta)\u0,
\end{aligned}
\end{equation}
\noindent where $S = D_u^{-\frac{1}{2}} W D_p^{-\frac{1}{2}}$, the symmetric normalization of weight matrix $W$.  We call this set of update rules the \textbf{\textit{BiRank iteration}}, which forms the core of the iterative BiRank algorithm.  In a nutshell, BiRank first randomly initializes the ranking vector, and then iteratively executes the BiRank iteration until convergence (summarized in Algorithm~\ref{alg:birank}).

For convergence, one can either monitor the change of ranking vectors $\pp, \uu$ across iterations, 
or rely on a held-out validation data to prevent overfitting. Moreover, the numerical convergence of BiRank is theoretically guaranteed, discussed later in Section~\ref{ss:convergence}. 

\begin{algorithm}[t]
	\caption{The Iterative BiRank Algorithm \label{alg:birank}}
	\KwIn{Weight matrix $W$, query vector $\p0, \u0$, and hyper-parameters $\alpha, \beta$;}
	\KwOut{Ranking vectors $\pp, \uu$;}
	Symmetrically normalize $W$:\quad $S = D_u^{-\frac{1}{2}} W D_p^{-\frac{1}{2}}$\;
	Randomly initialize $\pp$ and $\uu$\;
	\While{Stopping criteria is not met} { 
		$\pp \leftarrow \alpha S^T \uu + (1-\alpha)\p0$\;
		$\uu \leftarrow \beta S \pp + (1-\beta)\u0$\;
	}
	\Return{$\pp$ and $\uu$}
\end{algorithm}

\subsubsection{Time Complexity Analysis}
\label{ss:time_complexity}
It is easy to show that a direct implementation of BiRank iteration in Eq.~(\ref{eq:birank2}) has a time complexity of $O(|P|\cdot|U|)$, mainly due to the multiplication of $S^{T}\uu$ and $S\pp$. However, note that in real-world applications, the matrix $S$ is typically very sparse; for example in recommender systems, the user--item matrix to model is always over $99\%$ sparse (\eg Netflix challenge dataset).
In this case, a representation of sparse matrix only needs to account for non-zero entries~(which correspond to the edges of the bipartite graph), instead of all $|P|\cdot|U|$ entries. As such, the real-time cost needed for BiRank is $O(c |E|)$, where $c$ denotes the number of iterations executed to converge, and $|E|$ denotes number of edges in the graph. Thus, BiRank is linear with respect to number of edges, ensuring good scalability to large-scale graphs. Moreover, our empirical experience show that BiRank has a very fast convergence rate --- 10 iterations are usually enough for convergence. One reason is that it can be seen as optimizing a convex function effectively using alternating optimization, discussed later in Section~\ref{sec:regu}. 


\subsection{Convergence Analysis of BiRank}
\label{ss:convergence}
We show that BiRank can converge to a stationary and unique solution regardless of the initialization, followed by a theoretical analysis of the convergence speed. 
\newtheorem{theorem}{Theorem}
\newtheorem{lemma}{Lemma}

\subsubsection{Proof of Convergence}
\label{ss:convergence_proof}
It is clear that the behavior of BiRank depends on the hyper-parameters $\alpha$ and $\beta$, which are in the range $[0,1]$. To make a through analysis, we need to carefully consider the boundary conditions. Considering the two boundaries 0 and 1, we divide the proof into the following three cases: 

\begin{proof}
	\noindent\textbf{1. $\alpha=0$ or $\beta=0$.} 
	When $\alpha=0$, the vector $\pp = \p0$ is unchanged across iterations. Thus $\uu$, which depends on $\pp$ and $\u0$, will also be unchanged after the first iteration. Similarly for the case of $\beta=0$. \\
	
	\noindent\textbf{2. $\alpha=1$ and $\beta=1$.}
	In this case, the query vectors do not have any impact on the ranking, and the ranking is solely determined by the graph structure. The iterative update rule then reduces to Eq.~(\ref{eq:simple1}), whose matrix form is $\pp = S^T \uu, \uu = S \pp.$
	By further reducing this, we obtain:
	\begin{equation}
	\begin{aligned}
	\textbf{p}^{(k)} &= (S^T S)\textbf{p}^{(k-1)} = ... = (S^T S)^{k}\textbf{p}^{(0)}, \\
	\textbf{u}^{(k)} &= (S S^T)\textbf{u}^{(k-1)} = ... = (S S^T)^{k}\textbf{u}^{(0)},
	\end{aligned}
	\end{equation}
	where $k$ denotes the number of iterations, and $\textbf{p}^{(0)}$, $\textbf{u}^{(0)}$ denote the initial ranking scores for vertices. Note that matrix $S^T S$ and $S S^T$ are both symmetric matrices. According to a lemma in standard linear algebra~\cite{book:matrix}: \vspace{-5pt}
	\begin{lemma}
		If $M$ is a symmetric matrix, and $\textbf{v}$ is a vector not orthogonal to the principle eigenvector of $M$, then the limit of $M^{k}\textbf{v}$ (after scaling to a unit vector) converges to the principle eigenvector of $M$ with $k$ increasing without bound.
	\end{lemma}\vspace{-5pt}
	By the lemma, we can see that with reasonable initialization, the iterative process will converge to a stationary solution $\textbf{p}^*$ and $\textbf{u}^*$, which are the principle eigenvector of matrix  $S^T S$ and $S S^T$, respectively. \\
	
	\noindent\textbf{3. Normal cases.} We now consider the normal ranking scenarios that $\alpha$ and $\beta$ are in the range of (0,1), or one of $\alpha,\beta$ is 1, meaning that both the graph structure and query vectors can affect the ranking process. 
	Without loss of generality, we prove the convergence of $\textbf{p}$. 
	
	First, we eliminate $\textbf{u}$ in $\textbf{p}$'s update rule:
	\begin{equation}
	\label{eq:reduce_u}
	\pp = \alpha \beta (S^T S) \pp + \alpha(1-\beta)S^{T}\u0 + (1-\alpha)\p0. 
	\end{equation}
	
	Let matrix $M$ be $\alpha \beta (S^T S)$ and vector $\textbf{z}_0$ be $\alpha(1-\beta)S^{T}\u0 + (1-\alpha)\p0$, which are both invariant across iterations. Then, we have:
	\vspace{-3pt}
	\begin{equation}
	\begin{aligned}
	\label{eq:reduce_iter}
	\textbf{p}^{(k)} = M \textbf{p}^{(k-1)} + \textbf{z}_0 = ... = M^k \textbf{p}^{(0)} + \sum_{t=0}^{k-1} M^t \textbf{z}_0,
	\end{aligned}
	\end{equation}
	\noindent where $k$ denotes the number of iterations, and $\textbf{p}^{(0)}$ denotes the initial ranking vector of $\pp$. Assuming $M$'s eigenvalues are in the range of (-1, 1), we can obtain:
	\vspace{-3pt}
	\begin{equation*}
	\lim_{k \to \infty} M^k \textbf{p}^{(0)} = 0,\ and\  \lim_{k \to \infty}\sum_{t=0}^{k-1} M^t = (I - M)^{-1}.
	\end{equation*}
	\noindent where $I$ denotes the identity matrix. In this case, we can derive the stationary solution of $\pp$ as:
	\begin{equation}
	\label{eq:stat_p}
	\textbf{p}^* = (I - M)^{-1} \textbf{z}_0.
	\end{equation}
	However, the above stationary solution is derived based on the assumption that $M$'s eigenvalues are in the range (-1,1).  Next, we prove the correctness of this assumption. 
	\begin{theorem}
		\label{theorem1}
		$M$'s eigenvalues are in the range $[-\alpha\beta, \alpha\beta]$.
	\end{theorem}
	\vspace{-5pt}
	\begin{proof}
		Recall that $M$ is defined as:
		\begin{equation*}
		\begin{aligned}
		\label{eq:M}
		M = \alpha \beta (S^T S) &= \alpha \beta (D_u^{-\frac{1}{2}} W D_p^{-\frac{1}{2}})^T (D_u^{-\frac{1}{2}} W D_p^{-\frac{1}{2}}) \\
		&= \alpha \beta (D_p^{-\frac{1}{2}} W^T D_u^{-1} W  D_p^{-\frac{1}{2}}).
		\end{aligned}
		\end{equation*}
		To see $M$'s eigenvalues are bounded by $\alpha\beta$, we first construct a variant $M_v$ that has the same eigenvalues\footnote{The equality of the eigenvalues is easily shown by using determinants, denoted as $|\cdot|$. Let the eigenvalues of $M_v$ be $\lambda_v$, then we have $|M_v - \lambda_v I| = |D_p^{\frac{1}{2}} (M - \lambda_v I) D_p^{-\frac{1}{2}}| = |D_p^{\frac{1}{2}}|\cdot|M - \lambda_v I|\cdot|D_p^{-\frac{1}{2}}| = |M - \lambda_v I| = 0$, meaning that $\lambda_v$ are also $M$'s eigenvalues.} as $M$:
		\begin{equation*}
		\label{eq:M_prime}
		M_v = D_p^{\frac{1}{2}} M D_p^{-\frac{1}{2}} = \alpha\beta (W^T D_u^{-1} W  D_p^{-1}).
		\end{equation*}
		Note that matrix $W^T D_u^{-1} W  D_p^{-1}$ is a stochastic matrix in which the entries of each column sum to 1. By standard linear algebra~\cite{book:matrix}, for a stochastic matrix, its largest absolute value of the eigenvalues is always 1. Thus, the eigenvalues of $M_v$ are in the range $[-\alpha\beta, \alpha\beta]$, and same must hold for $M$ as they have exactly the same eigenvalues. The proof of $M's$ eigenvalues is finished.
	\end{proof}
	As $M$'s eigenvalues are theoretically guaranteed in the range $[-\alpha\beta, \alpha\beta]$ and in the normal cases $\alpha$, $\beta$ are in the range $(0,1)$, the assumption that $M$'s eigenvalues are in the range $(-1,1)$ holds. Therefore, we conclude that Eq.~(\ref{eq:stat_p}) indeed forms the stationary solution of $\textbf{p}$. To round out the proof, we give the stationary solution of BiRank as follows:
	\begin{equation}
	\begin{aligned}
	\label{eq:stationary}
	\textbf{p}^* &= (I - \alpha \beta S^T S)^{-1} [\alpha(1-\beta)S^{T}\u0 + (1-\alpha)\p0], \\
	\textbf{u}^* &= (I - \alpha \beta S S^T)^{-1} [\beta(1-\alpha)S\p0 + (1-\beta)\u0].
	\end{aligned}
	\end{equation} 
	The convergence proof of BiRank is finished.
\end{proof}

This convergence proof gives an elegant closed-form solution -- for any non-trivial initializations, the iterative algorithm of BiRank will converge to Eq.~(\ref{eq:stationary}). As such, an alternative method to our iterative BiRank is to direct calculate the closed-form stationary solution. 
Even so, we suggest the practitioners following the iterative procedure for two reasons. First, in real-world practice, when there is a large number of vertices to rank, the matrix inversion operation is very expensive, making the calculation of the closed-form solution inefficient. More specifically, matrix inversion is usually assumed to take $O(N^3)$ time~\cite{He:SIGIR16}; thus the time complexity of directly calculating the stationary solution is $O(|P|^3 + |U|^3)$, which even can be much higher than the upper bound of the iterative solution $O(c|P||U|)$.
Second, the iterative process emphasizes the underlying motivation that reinforcing a vertex's importance from its neighbors and the query vector. As such, one does not have to run the iterations until convergence; instead, one can compute the scores by starting from any initialization and performing a fixed number of iterations. 

\subsubsection{Speed of Convergence}
\label{ss:convergence_speed}

Since the behavior of BiRank depends on the graph structure, query vector and hyper-parameters $\alpha, \beta$, we analyze how do these factors impact BiRank's convergence speed. 

In each BiRank iteration, the score of a vertex comes from both its neighbors and the query vector. Since the query vector is static that remains unchanged across iterations, it cannot contribute to any form of divergence in the rankings; thus the main uncertainty for convergence stems from the part of the score diffusion from neighbors. As such, the number of iterations required to converge will increase as $\alpha$ and $\beta$ increase (the empirical evidence in Figure~\ref{fig:convergeRate} also verifies this property). Clearly, the slowest convergence is when $\alpha$ and $\beta$ are set to 1, where the effect of query vector is eliminated. 
When both $\alpha$ and $\beta$ are set to 1, the update of $\pp$ at the iteration $k$ can be written as:
$\textbf{p}^{(k)} = (S^T S)\textbf{p}^{(k-1)}$, 
which essentially can be seen as the power method for the symmetric eigenvalue problem (Chapter 8.2 of \cite{book:matrix}). It is known that the convergence of power method is determined by the second dominant eigenvalue of the transition matrix. In spite of the slight difference that the power iteration requires an additional $L_2$ normalization on the ranking vector (while our BiRank does not), we point out that BiRank shares the same property of convergence speed. 
\begin{theorem}
	\label{theorem2}
	The convergence rate of BiRank depends on the second largest eigenvalue of the matrix $S^T S$ in magnitude. 
\end{theorem}\vspace{-5pt}
\begin{proof}
	As $S^T S$ is a symmetric matrix, it is guaranteed to have $n$ eigenvalues which are real numbers ($n=|P|$). Let its eigenvalues be $\lambda_1, \lambda_2, ..., \lambda_n$, where $|\lambda_1|\ge |\lambda_2|\ge ... \ge|\lambda_n|$, and vectors $\textbf{x}_1, \textbf{x}_2, ..., \textbf{x}_n$ be the corresponding eigenvectors. Then, the starting vector $\textbf{p}^{(0)}$ can be expressed as:
	$\textbf{p}^{(0)} = \sum_{i=1}^n c_i \textbf{x}_i,$
	where $\{c_i\}$ are constant coefficients. Then the update of $\textbf{p}^{(k)}$ can be written as:
	\begin{equation}
	\begin{aligned}\textbf{p}^{(k)} &= c_1(S^T S)^k \textbf{x}_1 + c_2(S^T S)^k \textbf{x}_2 + ... + c_n(S^T S)^k \textbf{x}_n \\
	&= c_1 \lambda_1^k (x_1 + \sum_{i=2}^n c_i (\lambda_i/\lambda_1)^k \textbf{x}_i).
	\end{aligned}
	\end{equation}
	Here we use the fact that $(S^T S)\textbf{x}_i = \lambda_i\textbf{x}_i$. Hence, we see that the non-essential quantities decay at a rate of approximately $|\lambda_2/\lambda_1|$. As we have shown in Theorem~\ref{theorem1}, $S^T S$ has the same eigenvalues with a variant stochastic matrix, thus we have $|\lambda_1| = 1$. The proof is finished. 
\end{proof}
To summarize, the convergence rate of BiRank depends on the normalized adjacency matrix $S$ and parameters $\alpha,\beta$. Analytically, larger $\alpha$ and $\beta$ will lead to slower convergence; theoretically, smaller magnitude of the second dominant eigenvalue of $S^T S$ will result in faster convergence. 
In many applications, for high dimensional but sparse relational data (\eg user behaviors, documents), $S$ is usually of low rank. As a result, $|\lambda_2|$ is a small number, leading to a fast convergence of our BiRank algorithm. 

\subsection{Connection with Other Algorithms}
\label{ss:connection}
There are some bipartite graph ranking algorithms~\cite{Co-HITS:2009, rui2007bipartite, Cao2011} that share a similar spirit with BiRank, though originally developed for specific applications with varying ranking targets.  Specifically, in terms of the iterative ranking process, they have the same update rule form as Eq.~(\ref{eq:birank2}); the main difference is in how to generate the transition matrices ($S$ and $S^T$ for updating $\uu$ and $\pp$, respectively). 
It is instructive to clarify the difference with these algorithms. 
\vspace{-5pt}
\begin{table}[H]
	\begin{center}
		\caption{Transition matrices (\ie $S$ and $S^T$ in Eq.~(\ref{eq:birank2})) of different bipartite graph ranking algorithms. Note that $S^T$ here denotes a matrix, rather than just the transpose of $S$.}
		\vspace{-8pt}
		\label{tab:algo_difference}
		\begin{tabular}{ | l | l | }
			\hline
			\textbf{Method} & \textbf{Definition of Transition Matrices} \\ \hline
			HITS~(Kleinberg \cite{HITS:1999}) & $S = W; S^T = W^T$ \\ \hline
			Co-HITS~(Deng \etal \cite{Co-HITS:2009})	&  $S = W D_p^{-1}; S^T = W^T D_u^{-1}$ \\ \hline
			BGER~(Cao \etal \cite{Cao2011}) & $S = D_u^{-1} W; S^T = D_p^{-1} W^T$ \\ \hline
			BGRM~(Rui \etal \cite{rui2007bipartite}) & $S = D_u^{-1} W D_p^{-1}; S^T = D_p^{-1} W^T D_u^{-1}$ \\ \hline
			BiRank (our proposal) & $S = D_u^{-\frac{1}{2}} W D_p^{-\frac{1}{2}}; S^T = D_p^{-\frac{1}{2}} W^T D_u^{-\frac{1}{2}} $ \\ \hline
		\end{tabular}
	\end{center}
\end{table}
\vspace{-5pt}

Table~\ref{tab:algo_difference} summarizes the ways of constructing transition matrices using our symbol notation. From a high-level view, these algorithms differ in how they utilize the vertex degree to normalize each edge weight (except that HITS does not account for the query vector). 
HITS, the earliest proposed algorithm, uses the original weight matrix $W$ as-is; although the convergence can be guaranteed in theory, HITS is sensitive to outliers in graph~\cite{ng2001stable} and suffers from the tightly knit communities phenomenon~\cite{SALSA:2000}.
Co-HITS~\cite{Co-HITS:2009} normalizes each column of $W$ (and $W^T$) stochastically, having an explanation of simulating random walks on the graph. 
However, random walk methods can be biased towards the  high-degree vertices~\cite{Baluja:2008}.
While BGER~\cite{Cao2011} avoids this defect by normalizing each row of $W$ (and $W^T$) stochastically, yielding an effect of suppressing the scores of high-degree vertices. However, the one-side normalization of BGER does not account the degrees of $\pp$ vertices when updating $\uu$, allowing high-degree $\pp$ vertices to exert a stronger impact in the diffusion process; and vice versa. 
Similar with our proposed BiRank, BGRM also applies a symmetric normalization on $W$, while the level of normalization differs (the sum of normalization exponents is $-2$ and $-1$ for BGRM and BiRank, respectively). Although it is difficult to tell which way between them is more advantageous, we point out that BiRank 
employs the matrix $S^T S$ in a similar fashion to a stochastic matrix (the same eigenvalues, see Theorem~\ref{theorem1}) and corresponds to a regularization framework, both of which are nice properties that BGRM lacks. 


\section{Foundations of BiRank}
\label{sec:regu}
In contrast to the traditional graph ranking algorithms (\eg PageRank and HITS), BiRank iterations are implicitly optimizing an objective function.  This is analogous to the manifold ranking algorithm on graphs~\cite{Zhou04learningwith}.
In what follows, we investigate the regularization framework for BiRank and present a Bayesian explanation of the ranking algorithm. These two views shed important insight into the basis of BiRank, allowing future extensions in a theoretically principled way. To show its extendability, we finally generalize the methodology to rank for the more general \textit{n}-partite graphs.

\subsection{Regularization Framework}
Inspired from the discrete graph theory~\cite{zhou:2005}, we construct the regularization function as follows:
\begin{equation}
\begin{aligned}
\small
\label{eq:regu}
R(\textbf{u},\textbf{p}) &= \sum_{j=1}^{|P|}\sum_{i=1}^{|U|} w_{ij}(\frac{p_j}{\sqrt{d_j}} - \frac{u_i}{\sqrt{d_i}})^2 \\
&+ \gamma \sum_{j=1}^{|P|}(p_j-p_j^0)^2 + \eta\sum_{i=1}^{|U|}(u_i-u_i^0)^2,
\end{aligned}
\end{equation}
where $\gamma$ and $\eta$ are the regularization parameters to combine different components~(they are constants corresponding to $\alpha$ and $\beta$ in BiRank). 
Next, we first show that optimizing Eq.~(\ref{eq:regu}) leads to the iterative BiRank algorithm, and then interpret the meaning of the regularization function.

\subsubsection{Relationship with BiRank}
Eq.~(\ref{eq:regu}) defines an objective function with the ranking scores as model parameters. To optimize the objective function, a common strategy is performing the coordinate descent~\cite{Zhang:2016:DCF}. Let us first calculate its first-order derivatives with respect to each model parameter:
\vspace{-5pt}
\begin{equation}
\label{eq:derives}
\small
\begin{aligned}
\frac{\partial R}{\partial p_j} &= (2+2\gamma)p_j - 2\gamma p_j^0 - 2\sum_{i=1}^{|U|}\frac{w_{ij}u_i}{\sqrt{d_i}\sqrt{d_j}}\\
\frac{\partial R}{\partial u_i} &= (2+2\eta)u_i - 2\eta u_i^0 - 2\sum_{j=1}^{|P|}\frac{w_{ij}p_j}{\sqrt{d_i}\sqrt{d_j}}.
\end{aligned}
\end{equation}
\vspace{-5pt}

Setting the derivatives to 0, we can obtain:
\vspace{-5pt}
\begin{equation}
\label{eq:derives_0}
\small
\begin{aligned}
p_j &= \frac{1}{1+\gamma}\sum_{i=1}^{|U|}\frac{w_{ij}}{\sqrt{d_i}\sqrt{d_j}} u_i + \frac{\gamma}{1+\gamma} p_j^0, \\
u_i &= \frac{1}{1+\eta}\sum_{j=1}^{|P|}\frac{w_{ij}}{\sqrt{d_i}\sqrt{d_j}} p_j + \frac{\eta}{1+\eta} u_i^0,
\end{aligned}
\end{equation}
\noindent which exactly recovers the BiRank iteration Eq.~(\ref{eq:birank1}) by plugging $\gamma=\frac{1-\alpha}{\alpha}, \eta=\frac{1-\beta}{\beta}$ into the equation. As such, we see that BiRank is actually iterating towards optimizing the regularization function Eq.~(\ref{eq:regu}). 

As we have shown BiRank converges to a stationary solution in Section~\ref{ss:convergence_proof}. Now a question arises: does the solution found by BiRank lead to the regularization function's global optimum?  In fact it does, as the regularization function is strictly convex in both $p_j$ and $u_i$. 
\begin{theorem}
	The regularization function $R(\textbf{u},\textbf{p})$ defined by Eq.(\ref{eq:regu}) is strictly convex and only one global minimum exists.
\end{theorem}
\begin{proof}
	According to convex optimization theory, 
	a continuous, twice differentiable function is strictly convex if and only if its Hessian matrix is positive definite. As  $R(\textbf{u},\textbf{p})$ is a continuous function, we now prove it is twice differentiable and that its Hessian matrix is positive definite. 
	
	The second order derivative of $R(\textbf{u},\textbf{p})$ is: 
	\begin{equation*}
	\label{equ_second}
	\small
	\frac{\partial^2 R}{\partial p_j\partial p_j} = 2+2\gamma;\ 
	\frac{\partial^2 R}{\partial u_i\partial u_i} = 2+2\eta;\ 
	\frac{\partial^2 R}{\partial p_j\partial u_i} = 2\frac{-w_{ij}}{\sqrt{d_i}\sqrt{d_j}}.
	\end{equation*}
	We can see $R(\textbf{u},\textbf{p})$ is twice differentiable.
	
	Let matrix $A$ be the $(|U|+|P|)\times (|U|+|P|)$ weighted adjacency matrix of the bipartite graph. Then the Hessian of $R(\textbf{u},\textbf{p})$ can be written as:
	$
	2(I - D^{-\frac{1}{2}} A D^{-\frac{1}{2}}) + 2B,
	$
	where $D$ is a diagonal matrix where each entry $D_{kk}$ denotes the weighted degree of $k$-th vertex~(can be of either side); $B$ is a diagonal matrix that each entry $B_{kk}$ is $\gamma$ or $\eta$, dependant on the choice of origin (side) for the $k$-th vertex.
	
	Note that the matrix $(I - D^{-\frac{1}{2}} A D^{-\frac{1}{2}})$ is the normalized Laplacian matrix of the graph. By spectral graph theory, 
	the normalized Laplacian matrix of a graph is positive semi-definite. Meanwhile, $B$ is also positive definite because its eigenvalues are all positive~(eigenvalues of a diagonal matrix are its diagonal values).  Finally, according to the standard linear algebra, the addition of a positive semi-definite matrix and positive definite matrix is also positive definite. Thus, we reach the conclusion that the Hessian matrix must be positive definite. The proof is finished.
\end{proof}


\subsubsection{Interpretation of Regularization}
\label{ss:interpret}
It is instructive to interpret the meaning of the regularization function and see how it is constructed. First, it can be seen as enforcing two constraints in assigning the ranking scores on the vertices: A {\it smoothness} constraint that implies structural consistency --- that nearby vertices should not vary much in their ranking scores; and a {\it fitting} constraint which encodes the query vector --- that the ranking should not overly deviate from prior belief.

\noindent\textbf{Smoothness.} The first term of Eq.~\ref{eq:regu} implements the smoothness constraint,
which constrains a vertex's normalized score to be similar to the normalized scores of its connected neighbors. Minimizing it leads to the simplified BiRank algorithm devised in Eq.~(\ref{eq:simple1}). 
Moreover, it can be seen as the squared sum ($L_2$ distance) of \textit{edge derivatives} on the graph, as introduced in graph regularization theory~\cite{zhou:2005}:
\begin{equation*}
\label{eq:derivative}
\left.\dfrac{\partial f}{\partial e}\right|_{e_{ij}} = \sqrt{\dfrac{w_{ij}}{d_i}}u_i - \sqrt{\dfrac{w_{ij}}{d_j}}p_j,
\end{equation*}
which measures the variation (or energy drop) of the ranking function on edge $e_{ij}$.  {\it I.e.}, if two vertices are strongly connected but exhibit a large difference in their scores, then the magnitude of the variation will be large. Variants of our vanilla BiRank can be derived by employing other methods to combine the edge derivatives, \eg the $L_1$ distance, which can yield and model different effects for smoothness.

\noindent\textbf{Fitting.} The second and third term of the regularization function gives the fitting constraint for the query vectors $\p0$ and $\u0$, respectively:
\begin{equation}
\label{eq:fitting}
R_f(\pp)=\sum_{j=1}^{|P|}(p_j-p_j^0)^2,\quad 
R_f(\uu)=\sum_{i=1}^{|U|}(u_i-u_i^0)^2,
\end{equation}
This fitting term is easy to understand: it regularizes the value of each vertex's score to be similar with its prior score, \ie its value in the query vector.

In our formulation of BiRank, we have chosen a MSE~(mean squared error) loss function form; other ranking-oriented loss functions, such as the BPR-OPT~\cite{BPR:2009}, may be more suited if one seeks to maintain the vertices' relative ordering in the query vector during the ranking process. We leave this possibility for future work. 

\subsection{Bayesian Explanation}
\label{ss:bayesian}

On the basis of the above regularization framework, we now present a Bayesian explanation for BiRank. 

\begin{figure}[t]
	\centering
	\includegraphics[width=0.25\textwidth]{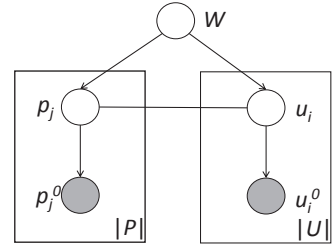}
	\vspace{-5pt}
	\caption{Graphical model representation of BiRank.}
	\label{fig:bayesian}
	\vspace{-10pt}
\end{figure}

Figure~\ref{fig:bayesian} shows the graphical model representation of the ranking method.
We model the query vectors $\p0$ and $\u0$ as observations, which are generated by the latent factors $\textbf{p}$ and $\textbf{u}$ (distributions), serving as the importance scores of the vertices; the weight matrix $W$ forms the prior for generating the latent factors. The goal is to infer the latent factors $\textbf{p}$ and $\textbf{u}$ that generate the observations $\p0$ and $\u0$. 

The MAP~(maximum a posteriori) estimation is given by:
\begin{equation*}
\argmax_{\uu,\pp} p(\uu,\pp|\u0,\p0, W).
\end{equation*}
By Bayes' rule and the conditional independence indicated in the graphical model, we have:
\begin{equation*}
\begin{aligned}
p(\uu,\pp|\u0,\p0,W) &= \frac{p(\u0,\p0|\uu,\pp)\cdot p(\uu,\pp|W)}{p(\u0,\p0)} \\
&\propto p(\u0,\p0|\uu,\pp)\cdot p(\uu,\pp|W) \\
&\propto p(\u0|\uu)\cdot p(\p0|\pp)\cdot p(\uu,\pp|W).
\end{aligned}
\end{equation*}
Note that $\pp$ and $\uu$ are not conditionally independent with each other given the prior $W$, as a vertex's score is also influenced by its neighbors' scores.
Taking the logarithm, MAP estimation is then equivalent to:
\begin{equation*}
\argmax_{\uu,\pp} \left\{
\ln p(\uu,\pp|W) + \ln p(\u0|\uu) + \ln p(\p0|\pp)
\right\}.
\end{equation*}

We then devise the conditional probabilities as follows:
\begin{equation*}
\begin{aligned}
p(\uu, \pp|W) &= \frac{1}{Z_{up}} e^{-R_s(\uu,\pp)}, \\
p(\p0|\pp) = \frac{1}{Z_p} e^{-\gamma R_f(\pp)}&, \qquad
p(\u0|\uu) = \frac{1}{Z_u} e^{-\eta R_f(\uu)}, 
\end{aligned}
\end{equation*}
where $Z_{up}$, $Z_p$ and $Z_u$ are normalization constants, and where $R_s(\uu,\pp)$ is the smoothness term, and $R_f(\pp)$ and $R_f(\uu)$ are the fitting terms defined in Eq.~(\ref{eq:fitting}). From this formalization, we can see that minimizing the regularization function is equivalent to maximizing the posteriori probability of generating the query vector. 

This shows the equivalence between BiRank's ranking process and a Bayesian network. We map the ranking problem to probabilistic graphical modeling, allowing the extension of BiRank in a probabilistic way, which is more flexible and adaptable for different applications. For example, if there is additional prior knowledge or context for the vertices, we can model them as priors of $\pp,\uu$ and use the desired distributions; moreover, aside from MAP, other inference techniques can also be applied to infer the ranking scores, such as variational inference and MCMC sampling. 

\subsection{Generalization to \textit{n}-partite Graph Ranking}
Our proposed BiRank methodology is general and versatile. Here, we generalize it to rank vertices of the more general \textit{n}-partite graphs.
A \textit{n}-partite graph is a graph whose vertices can be partitioned into $n$ different independent sets. We represent it as $G(\{P_t\}; \{E_{tl}\})$, where $P$ represents vertices, $E$ represents edges, $t$ and $l$ represent the indices of the independent vertex sets, satisfying $1\le t,l \le n$. Let the weight matrix of edges $E_{tl}$ be $W_{tl}$, which is a $|P_t|\times |P_l|$ matrix. If the graph is undirected, we have $W_{lt} = W_{tl}^T$. The symmetrically normalized matrix is defined as $S_{tl} = \sqrt{D^t} W_{tl} \sqrt{D^l}$, where $D^t$ and $D^l$ represent the diagonal degree matrix of vertices $P_t$ and $P_l$, respectively. Let the ranking vector and query vector of vertices in $P_t$ be $\textbf{p}_t$ and $\textbf{p}_t^0$, respectively. Then, the objective function for vertex ranking is defined by smoothing the connected vertices (of pairwise vertex types) and fitting the query vectors (of each vertex type):
\begin{equation*}
\small
R = \sum_t \eta_t ||\textbf{p}_t - \textbf{p}^t_0 ||^2 + \sum_{l\ne t} \gamma_{tl} \sum_{i,j} (W_{tl})_{ij} (\frac{(\textbf{p}_t)_i}{\sqrt{D^t_{ii}}} - \frac{(\textbf{p}_l)_j}{\sqrt{D^l_{jj}}})^2,
\end{equation*}
where $\gamma_{tl}$ and $\eta_{t}$ are hyper-parameters that control the importance of the corresponding component. Similar to BiRank, this regularization function is strictly convex for all model parameters. Thereby, the global minimum can be achieved by alternating optimization, which leads to the iterative update solution as:
\begin{equation}
\label{eq:n-partiteRank}
\small
\begin{aligned}
\textbf{p}_1 &= \sum_{l\ne 1} \alpha_{1l} S_{1l} \textbf{p}_l + (1 - \sum_{l\ne 1} \alpha_{1l}) \textbf{p}_1^0, \\
&...... \\ 
\textbf{p}_t &= \sum_{l\ne t} \alpha_{tl} S_{tl} \textbf{p}_l + (1 - \sum_{l\ne t} \alpha_{tl}) \textbf{p}_t^0,  \\
\end{aligned}
\end{equation}
where the hyper-parameters $\alpha_{tl}$ are associated with $\gamma_{tl}$ and $\eta_{t}$, indicating the weight of graph substructure $E_{tl}$ in contributing to the final ranking. 
Iteratively executing the above update rule until convergence, we obtain the ranking. We call this algorithm \textit{n-partiteRank}, as a generalization of BiRank for \textit{n}-partite graphs; it is easy to see when $n=2$, the algorithm exactly recovers the BiRank. Also, the time complexity of the algorithm is linear to number of edges in the \textit{n}-partite graph, which is very efficient for large-scale heterogeneous graphs in real-world applications. 
In Section~\ref{ss:recom}, we demonstrate how to utilize this generic algorithm to model user reviews ($n=3$, \ie TriRank) for the application of personalized recommendation~\cite{He:2015}. 


\section{Applications}
\label{sec:applications}
In this section, we demonstrate how to apply the BiRank method to two real-world applications; namely, 1) predicting the future popularity of items, and 2) recommending items of interest to users. We choose to model user comment data for addressing the relevant task, since it is a form of explicit feedback that is easily accessible to both content providers and external observers\footnote{In contrast, implicit feedback --- such as users' clicks on webpages and views on items --- is only obtainable for the internal content providers. For external observers, such as the third-party services, implicit feedback is usually difficult to access.}.

\subsection{Popularity Prediction}
\label{ss:popularity}
Predicting the popularity of web content has a wide range of
applications, such as online marketing~\cite{Huberman:2010}
and recommender system~\cite{He:SIGIR16}. In what follows, we first briefly introduce the task, and then show how to customize BiRank to address the problem. 

\subsubsection{Task Introduction}
A direct and objective metric to measure an item's popularity is the view count, which evaluates users' attention on the item. 
Thereby, previous works have primarily focused on modeling the view histories of items~\cite{Pinto:2013, Huberman:2010} and casted the prediction as a regression task. However, for some external services (who are not the content providers themselves), items' view histories are not easily accessible. Specifically, most websites do not explicitly provide the view history for an item. Even in the cases where a website like YouTube and Flickr provides the current number of views, one will have to repeatedly and periodically crawl the item pages to build view histories, a very bandwidth intensive activity.

To assist external observers in predicting items' popularity, we are more interested in an alternative and more viable solution --- modeling the affiliated user comments of items. In contrast to view count, the advantage of comments is the exposure of users' commenting activities up to the current time --- crawling once, one can get the previous history and perform the prediction directly. However, the key deficiency is that the comment history can be much sparser than view history, since a user viewing an item may not comment on it.
For example, it is common that two items have no comments during the time interval, while they attract views at a different rate. 
As such, existing view-based solutions~\cite{Pinto:2013, Huberman:2010} (which are mostly regression-based methods) can fail to leverage the comment series to predict popularity accurately. 

To tackle the sparsity issue in user comments for quality prediction, we need to account for more popularity signals in addition to the comment count. Here, we propose three ranking hypotheses observed from user comments that we wish to incorporate into our popularity prediction solution:

\textbf{H1. Temporal Factor:} If an item has received many comments recently, it is more likely to be popular in the near future. More recent comments are a salient signal that more users focused on the item recently. 

\textbf{H2. User Social Influence:} If the users commenting on an item are more influential, the item is more likely to receive more views in the future. This is enabled by the Web~2.0 social interfaces that propagate a user's comments to friends and followers. 

\textbf{H3. Item Current Popularity:} If an item has already been popular, 
it is likely to garner more views in the future. This is partially effected by the existing visual interfaces of Web 2.0 systems: the more views an item has, the more likely it will be promoted to users. 

\subsubsection{BiRank Customization}
To customize BiRank for a certain ranking purpose, we need to construct the weighted bipartite graph model and devise the query vectors. 

\noindent\textbf{Bipartite Graph Construction}.
As we deal with users' commenting behaviors on items, we model their relationship as a bipartite graph -- users and items form the two sides of vertices $U$ and $P$, respectively, and edges $E$ represent comments. If and only if a user has commented on an item, there is an edge between them. 
We use the edge weight to model the respective comment's contribution towards the item's future popularity. As the hypothesis $H1$ shows a strong near-term correlation, we assign $w$ based on temporal considerations. Specifically, recent (older) comments should contribute more (less) to an item's future popularity. 
To achieve this, we choose a monotonically decreasing exponential decay function: 
\vspace{-3pt}
\begin{equation} \small
\label{eq:edge_weight}
w_{ij}=\delta^{a(t_0 - t_{ij})+b},
\end{equation}
where $\delta$ is the decay parameter that controls the decay rate, $t_0$ is the ranking time and $t_{ij}$ is the commenting time; $a$ and $b$ are constants, to be tuned for the particular media and site. 
Time units are arbitrary; they can be assigned as minutes, hours, days, weeks or other units, depending on the temporal resolution and the domain of the items to rank. If no edge exists between $u_i$ and $p_j$, then $w_{ij}$ is zero.
In our empirical study, we find a setting of $\delta=0.85, a=1, b=0$ leads to good performance, and thus use this setting across datasets. As we focus on short-term popularity prediction, we set the time unit as 1 day.

\noindent\textbf{Query Vector Setup}.
We devise the user query vector $\u0$ and item query vector $\p0$ to account for the hypotheses $H2$ and $H3$, respectively. 
Intuitively, if a user has more friends, his behavior is likely to influence more users. Thus we set a user's prior score in the query vector proportional to the log value of his number of friends:
\begin{equation*} \small
\label{eq:u0}
u_i^0 = \dfrac{\log (1+g_i)}{\sum_{k=1}^{|U|} \log (1+g_k)},
\end{equation*}
where $g_i$ is user $u_i$'s number of friends at the ranking time. Note that we use add-$1$ smoothing to address the case where a user has no friends. 

$H3$ models the current popularity factor on items. As such, the item query vector should encode our prior belief on each item's popularity prior to examining its recent comments.
We capture this potential ``rich-get-richer'' effect by defining an item's score in the query vector as:
\begin{equation*}
\label{eq:p0}
p_j^0 = \dfrac{\log v_j}{\sum_{k=1}^{|P|} \log v_k},
\end{equation*}
\noindent where $v_j$ denotes the view count of item $p_j$ at ranking time. \\

After finalizing the edge weights and query vectors, 
the rationale in our design can be more clearly seen by looking into the BiRank iteration in Eq.~(\ref{eq:birank2}). 
First, it captures the mutual reinforcement between users and items --- the more recent the comments are by a user on an item, the higher the popularity score the item will receive; and in return, the popularity of the target item increases the user's influence.  
Second, the score of items and users is partially determined by the original setting of the query vector. 
To sum up, BiRank determines a user's social influence based on two source of evidence: his level of activity and his number of friends.  Analogously, BiRank determines an item's future popularity based on four aspects: the frequency and recency of comments on it, the influence of the users commenting on it, and its current accumulated popularity. 
Thus, from a qualitative point of view, we see that the formulation of BiRank can encode our hypotheses on the ranking function. 
\vspace{-5pt}

\subsection{Personalized Recommendation}
\label{ss:recom}
In this subsection, we apply the generalized, $n$-partiteRank to the application of personalized recommendation. This is a more challenging task than popularity prediction, since it needs to generate a personalized ranking of items for each user. In what follows, we first show how to employ BiRank to encode the well-known collaborative filtering effect for recommendation. Then we use the TripartiteRank (short for TriRank) to additionally model aspects (extracted from comments' texts) for enhanced recommendation. 

\subsubsection{Collaborative Filtering with BiRank}
\label{ss:recom_birank}
In recommendation systems, collaborative filtering (CF) is the most successful and widely-used technique for personalization. 
It exploits user--item interactions~(\eg ratings, click histories) by assuming that users with similar interest consume similar items. The core of a CF algorithm lies in how it models the similarity among users and items. For example, neighbor-based CF~\cite{ItemKNN:WWW2001} directly estimates the similarity by adopting statistical measure on the user--item matrix, while latent factor-based CF~\cite{He:SIGIR16} estimates the similarity by projecting items/users into a latent space. Under our BiRank paradigm, similarity is estimated by means of smoothing the user--item graph with the target user's known preference, embodied as a query vector. We use an example in Figure~\ref{fig:toy_recom} to illustrate how the smoothness works. 
\vspace{-15pt}
\begin{figure}[h]
	\centering
	\begin{subfigure}[b]{0.1\textwidth}
		\centering
		\includegraphics[width=\textwidth]{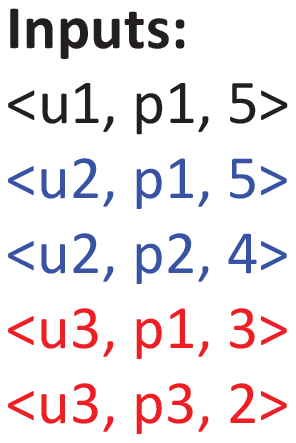}
	\end{subfigure} \hspace{+15pt}
	\begin{subfigure}[b]{0.20\textwidth}
		\centering
		\includegraphics[width=\textwidth]{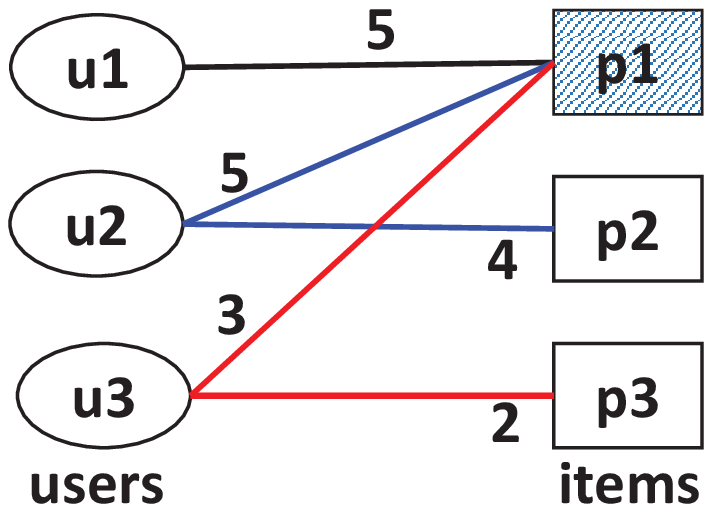}
	\end{subfigure} \vspace{-10pt}
	\caption{A toy example of using BiRank to model the collaborative filtering effect. The target user $u_1$ has previously rated item $p_1$ with a rating score 5 (in tail).}
	\label{fig:toy_recom}
	\vspace{-10pt}
\end{figure}

Users and items represent the two types of vertices, and edge weights denote the rating scores (here, a zero score means the user did not rate the item; a missing value). Assume we want to recommend items to the target user $u_1$, who has rated $p_1$ with a score of 5. We construct the query vector by setting the prior of $p_1$ to 5 and other vertices to 0. Now, we consider how the BiRank predicts $u_1$'s preference (\ie the similarity to other items) with this setting.
As $p_1$ is connected more strongly to $u_2$ than $u_3$, by the smoothness constraint, $u_2$ will be given a higher score than $u_3$. This indicates that BiRank treats $u_2$ more similar with $u_1$ than $u_3$. 
Finally, since the edge weights of $<u_2$,$p_2>$ and $<u_3$,$p_3>$ are identical, BiRank will assign $p_2$ a higher score than $p_3$, meaning that $p_2$ is a more suitable candidate to recommend for $u_1$ than $p_3$.
From this qualitative analysis, we see that by properly setting the query vector's values, smoothing the user--item relation results in a collaborative filtering effect.
More specifically, by setting the query vector as the rated items of the target user, BiRank functions similar to item-based CF~\cite{ItemKNN:WWW2001} which represents a user by his historical actions for personalization. 

\subsubsection{Modelling Aspects with TriRank}
Aside from ratings, which form the basis for collaborative filtering, most Web~2.0 systems also encourage users to pen reviews. These reviews justify a user's ratings, offering the underlying reasons for the rating by discussing the \textit{specific properties} of the item. We term these specific properties as \textit{aspects}, which are nouns or noun phrases that represent the features of items (see Table~\ref{tab:aspect_examples} as examples). Aspects are well-suited as a complementary data source for CF, since a mention of an aspect implies the user's interest in the aspect, which in turn reveals the user's preference. In this subsection, we model the aspects with TriRank to improve the CF-based recommendation. Similar to the application of BiRank to popularity prediction, we first show how to construct the tripartite graph, and then design the query vectors to implement the personalized ranking. 

\begin{table}[t]
	\begin{center}
		\caption{\textbf{Top automatically extracted aspects.}} 
		\label{tab:aspect_examples}
		\small
		\vspace{-8pt}
		\begin{tabular}{| l | C{6.5cm} | }
			\hline
			\textbf{Yelp}	& bar, salad, menu, chicken, sauce, restaurant, rice, cheese, fries, bread, sandwich, drinks  \\ \hline
			\textbf{Amazon}	& camera, quality, sound, price, battery, pictures, screen, size, memory, lens  \\ \hline
		\end{tabular}
		\vspace{-8pt}
	\end{center}
\end{table}

\textbf{Tripartite Graph Construction}.
After extracting aspects from user reviews, we construct a tripartite graph with users, items and aspects as the three types of vertices. We formalize the input as a list of triples, where each triple $<u_i, p_j, a_k>$ denotes that user $u_i$ has rated item $p_j$ with a review mentioning aspect $a_k$ and is represented as a triangle with edges $e_{ij}, e_{ik}$ and $e_{jk}$ in the graph. Figure~\ref{fig:toy_tripartite} shows an example of the tripartite graph.
\vspace{-10pt}
\begin{figure}[h]
	\centering
	\begin{subfigure}[b]{0.10\textwidth}
		\centering
		\includegraphics[width=\textwidth]{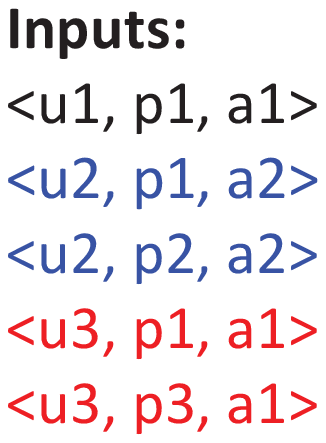}
	\end{subfigure} \hspace{+10pt}
	\begin{subfigure}[b]{0.25\textwidth}
		\centering
		\includegraphics[width=\textwidth]{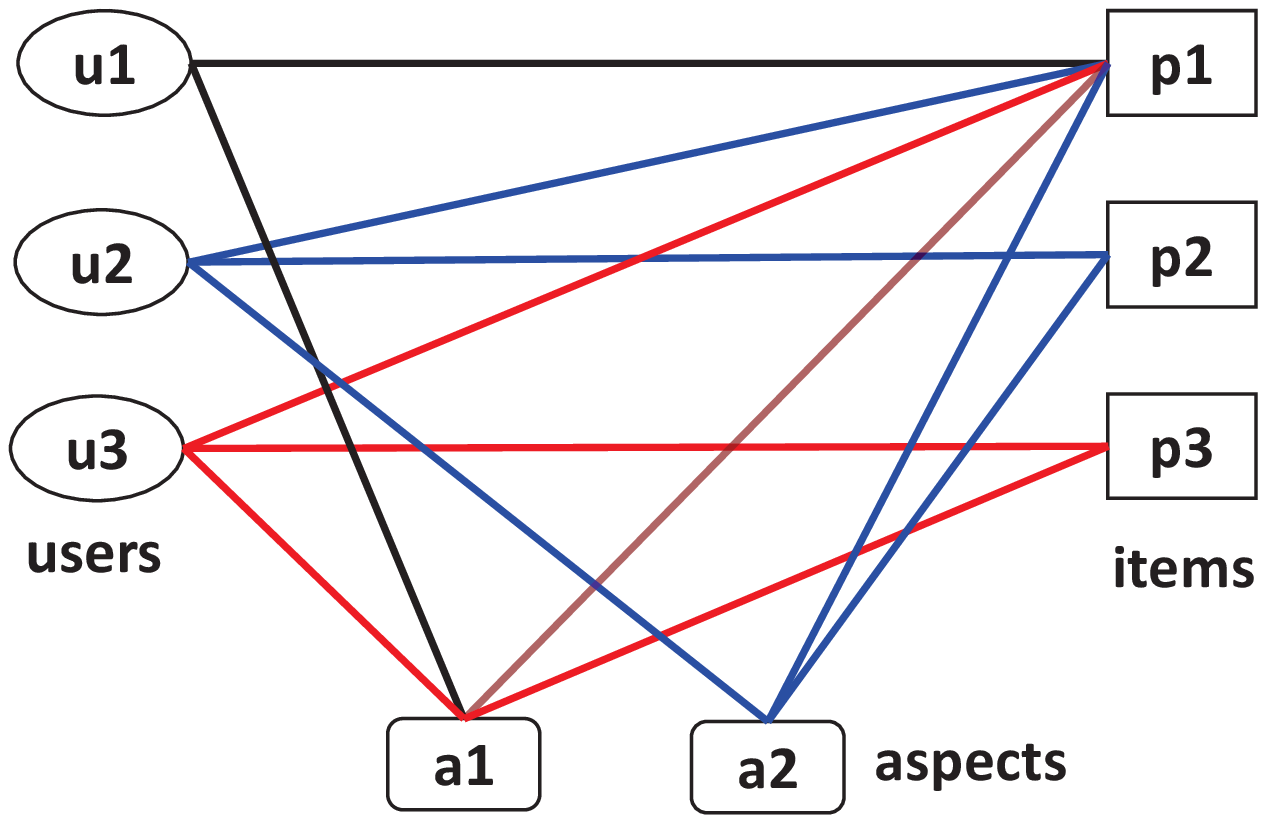}
	\end{subfigure}
	\caption{An example of the tripartite graph~(the dashed line illustrates the additional input $<u_1, p_3, a_2>$).}
	\vspace{-5pt}
	\label{fig:toy_tripartite}
\end{figure}

Each edge carries a weight, which is crucial to determine the meaning of smoothness and the behavior of TriRank. 
The setting of \textit{user--item} edge weights should encode the collaborative filtering effect: in cases with explicit feedback, it can be the rating score (as illustrated previously in Section~\ref{ss:recom_birank}); for implicit feedback, it can denote whether the user has interacted with or browsed the item (measured as either a binary yes/no, or an integer view count). Our datasets provide explicit user ratings, so we use these ratings as-is. 
The setting of aspect connected edges should reflect the aspect filtering effect: if a user is interested in an aspect, then the system should rank the items that are good at this aspect high. Thus, we set the edge weights of \textit{user--aspect} relation and \textit{item--aspect} relation to connote the degree of user interest (item specialty) with respect to the aspect. 
Once aspects are identified in reviews, we use the review frequency (\ie number of reviews that mention the aspect) within all a user's (item's) reviews as the edge weight. 
As is done in general information retrieval, we take the logarithm of the review frequency, to dampen the effect of aspects that appear very frequently.

\noindent\textbf{Query Vector Setup}.
The query vectors should encode the target user's prior preference on the vertices, which serve as the gateway for personalization. Here we discuss how to set the query vectors for target user $u_i$. 

For the item query vector $\textbf{p}_0$, an element takes a positive value if the target user has interacted with the item; otherwise, 0. Thus we adopt the $i^{th}$ row vector of the user--item matrix as the $\textbf{p}_0$ for the target user $u_i$. Similarly, the aspect query vector $\textbf{a}_0$ is set as the respective row vector of the user--aspect matrix, denoting the target user's prior preference on aspects. The user preference vector $\textbf{u}_0$ should denote the target user's similarity with other users. When user's social network is available, we can use her friends information to initialize $\textbf{u}_0$.  Due to the lack of social information in our dataset, we adopt a simple approach, setting the target user herself as 1, and all other users as 0. Considering that the weight matrix is symmetrically normalized, we also apply the $L_1$ norm on $\textbf{p}_0, \textbf{a}_0$ and $\textbf{u}_0$ respectively, for a meaningful combination. \\

Our final recommendation solution works as follows. After constructing the tripartite graph, we preform TriRank with the personalized query vector for each target user. The ranking process follows the iteration defined in Eq.~(\ref{eq:n-partiteRank}). After convergence (usually in 10 iterations), items with the highest scores serve as the recommendations for the user, and the aspect vertices with the highest scores can be used as explanatory factors for the recommendations~\cite{He:2015}. 

\section{Experiments}
\label{sec:exper}

In this section, we empirically examine BiRank's properties and effectiveness. We first conduct experiments on synthetic data to study BiRank's convergence and time efficiency. Then we perform experiments on real-world datasets to evaluate BiRank performance for the two applications of popularity prediction and personalized recommendation. 

\subsection{Experiments on Synthetic Data}

\subsubsection{Datasets}
We concern ourselves with two forms of generated graphs: 

\noindent\textbf{1. Synthetic Random Graphs.} These random graphs are generated by sampling edges from a uniform distribution. 
We control the density of generated bipartite graphs to simulate real-world graph sparsity.
Given the expected density of the graph, we visit each potential edge and generate a uniformly random number in the range $(0,1)$; if the number is less (or equal) than the density value, we add the edge into the graph. 

\noindent\textbf{2. Synthetic Power-law Graphs.} Considering that many real-world graphs follow a power-law distribution, such as document--word and user--item graphs,
we also generate the power-law bipartite graphs. We adopt the power-law graph generation algorithm in \cite{Gao:2013}: starting from an empty graph, it follows two main steps: first it assigns a degree $x$ to each vertex $v$ from the distribution $p(d_v = x)\propto x^{-\lambda}$ where $\lambda > 1$; then, it sorts the vertices by degree in decreasing order, and assigns neighbors to each vertex according to the degree. We adjust the second step for generating bipartite graph -- sampling neighbors of a vertex only from the vertices of the other side. 

%

\subsubsection{Convergence Study}
\label{ss:convergence_study}
There are two natural questions need to be answered empirically regarding BiRank's convergence\footnote{Note that due to the difficulty of controlling the eigenvalues of generated graphs, we do not empirically study the impact of the second dominant eigenvalue on convergence rate. While the impact has been theoretically proved in Theorem~\ref{theorem2}.}:
\begin{enumerate}
	\item Will BiRank iterations converge to the optimal solution of the regularization function as analyzed theoretically? 
	\item How does the algorithm hyper-parameters~(\ie $\alpha$ and $\beta$) influence BiRank's convergence rate? 
\end{enumerate}
Since our findings were consistent across many settings, we only report results with a set of representative settings.
\vspace{-10pt}
\begin{figure}[h]
	\centering
	\begin{subfigure}[b]{0.24\textwidth}
		\centering
		\includegraphics[width=\textwidth]{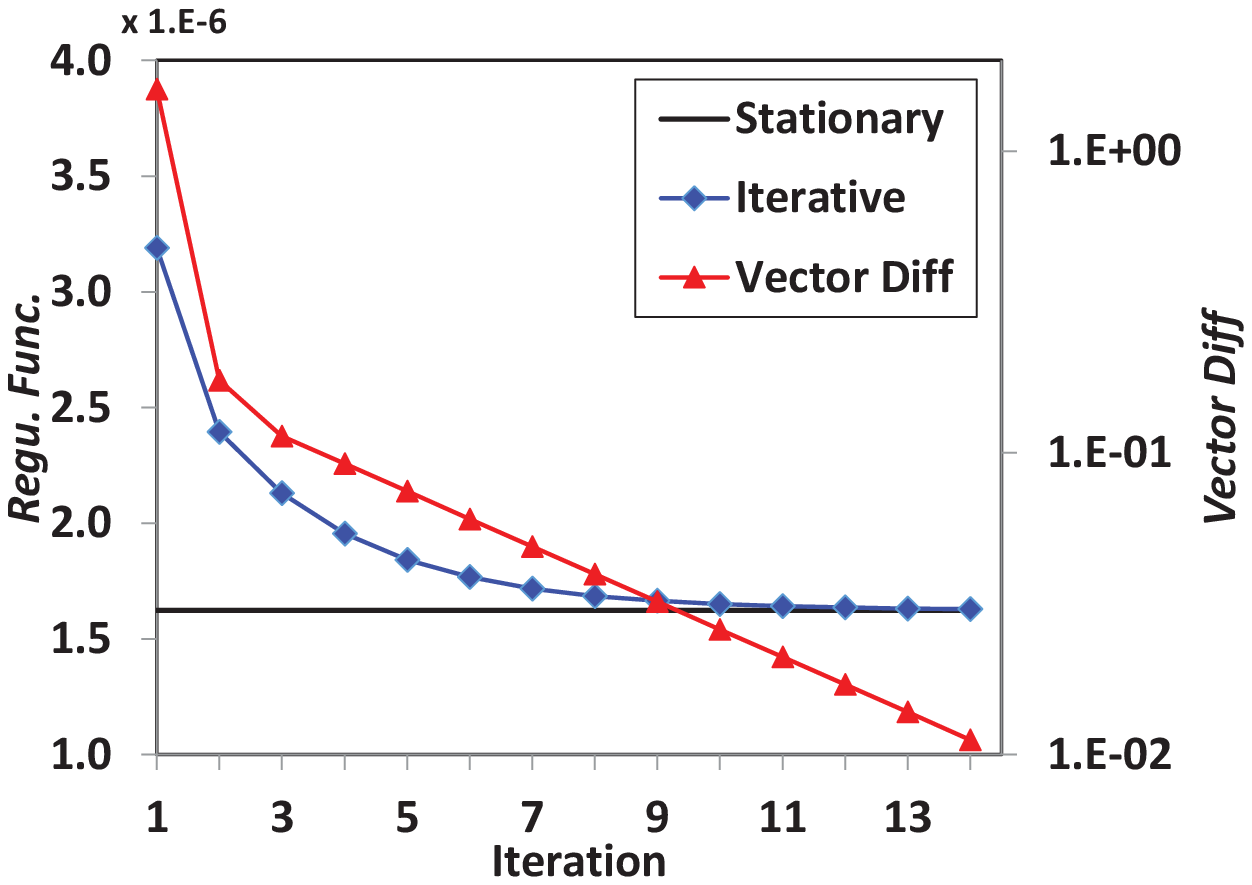} \vspace{-15pt}
		\caption{Random~($10K\times 50K$)}
	\end{subfigure} 
	\begin{subfigure}[b]{0.24\textwidth}
		\centering
		\includegraphics[width=\textwidth]{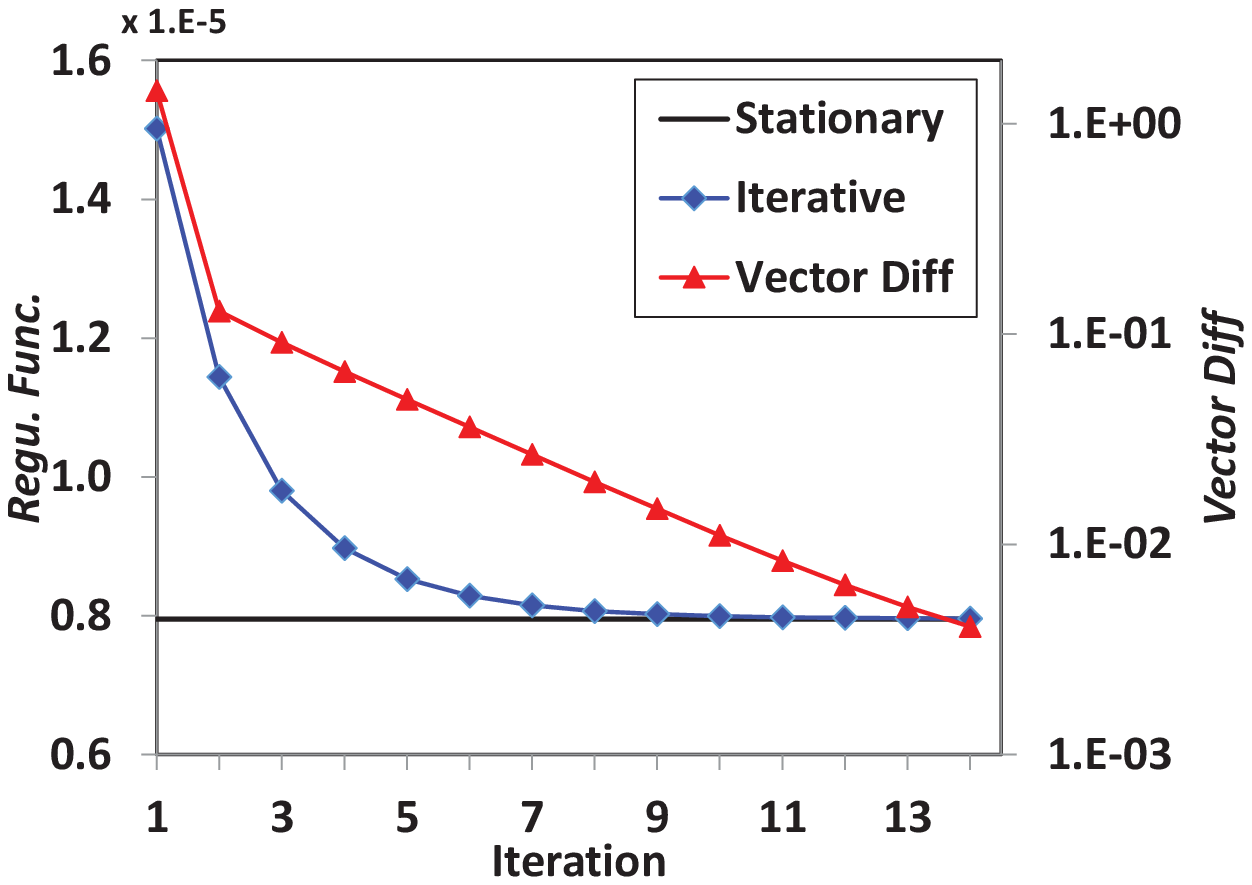} \vspace{-15pt}
		\caption{Power-law ($10K\times 50K$)}
	\end{subfigure} \vspace{-15pt}
	\caption{Convergence status of two generated graphs.}
	\label{fig:converge}
	\vspace{-5pt}
\end{figure}

\noindent\textbf{1. Convergence to optimum.}
Figure~\ref{fig:converge} plots the convergence status of each iteration on two representative synthetic graphs (the random setting has a density of  1\%; the power-law graph sets $\lambda=2$). 
The black line~(Stationary) benchmarks the optimal solution from the direct calculation of the stationary solution Eq.~(\ref{eq:stationary}). 
The blue line~(Iterative) shows the regularization function's value after the update of each iteration; the red line~(Vector Diff, y-axis scale of the right) shows the difference~(\ie squared sum) of the ranking vector before and after the update of each iteration. 
As we can see, BiRank successfully finds the optima of the regularization function in all four cases. Our further examination (not shown) validates that the ranking vector obtained by BiRank iterations is actually same with the stationary solution. This demonstrates BiRank's ability in converging to the unique and optimal solution of the regularization function, regardless of the graph structure.
Moreover, the convergence rate is rather fast for these simulated problems -- also usually within 10 iterations.
Another finding is that the deepest descents are in the early iterations, which impose the most influence to the ranking.

\noindent\textbf{2. Convergence rate \wrt algorithm parameters.} In BiRank, $\alpha$ and $\beta$ are the hyper-parameters to combine the score calculated from the graph structure and query vector. They act like the damping factor in PageRank, and are crucial to the ranking results and convergence. We study the impact of the two parameters on the convergence rate. The convergence threshold is set as 0.0001 for Vector Diff, a strict condition that guarantees a sufficient convergence.
Figure~\ref{fig:convergeRate} plots the number of iterations to converge on two graphs of size $10K\times 20K$. 
Both graphs show the same trend that BiRank needs more iterations to converge with a larger $\alpha$~(and $\beta$). 
This verifies our qualitative analysis in Section~\ref{ss:convergence_speed} that smaller value of $\alpha$~(and $\beta$) leads to a larger effect of the static query vectors, helpful in achieving quick convergence. 

\begin{figure}[t]
	\centering
	\begin{subfigure}[b]{0.20\textwidth}
		\centering
		\includegraphics[width=\textwidth]{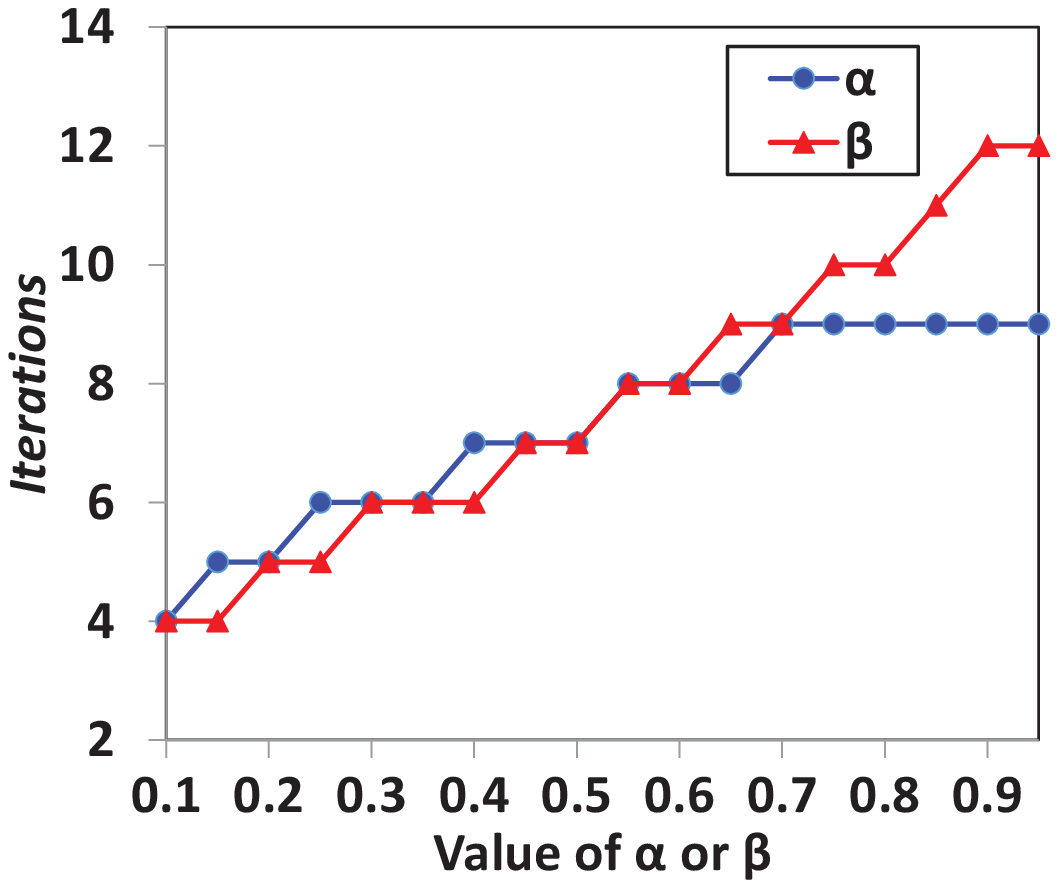} \vspace{-15pt}
		\caption{Rand (density=1\%)}
	\end{subfigure} \hspace{+10pt}
	\begin{subfigure}[b]{0.20\textwidth}
		\centering
		\includegraphics[width=\textwidth]{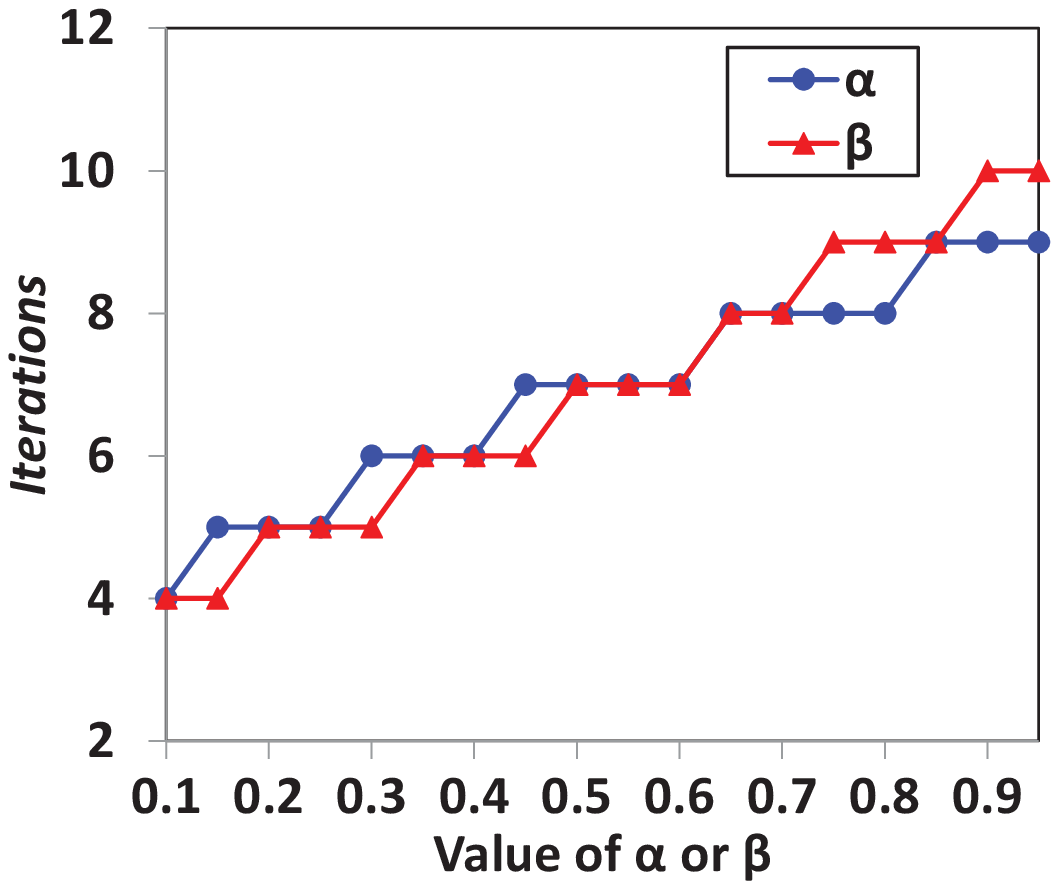} \vspace{-15pt}
		\caption{Power-law ($\lambda=2$)}
	\end{subfigure} \vspace{-5pt}
	\caption{Convergence rate \wrt $\alpha$ and $\beta$.}
	\vspace{-18pt}
	\label{fig:convergeRate}
\end{figure}

\vspace{-5pt}

\subsubsection{Time Efficiency}
In Section~\ref{ss:time_complexity}, we analyzed the theoretic time complexity of BiRank: $O(n)$, in the number of graph edges. We now empirically validate this property. We adopt sparse representation for matrices, and implement BiRank in Java. The experiments are run on a modern desktop (Intel 3.5GHz CPU with 16GB RAM running on a single thread).
\vspace{-10pt}
\begin{figure}[h]
	\centering
	\begin{subfigure}[b]{0.20\textwidth}
		\centering
		\includegraphics[width=\textwidth]{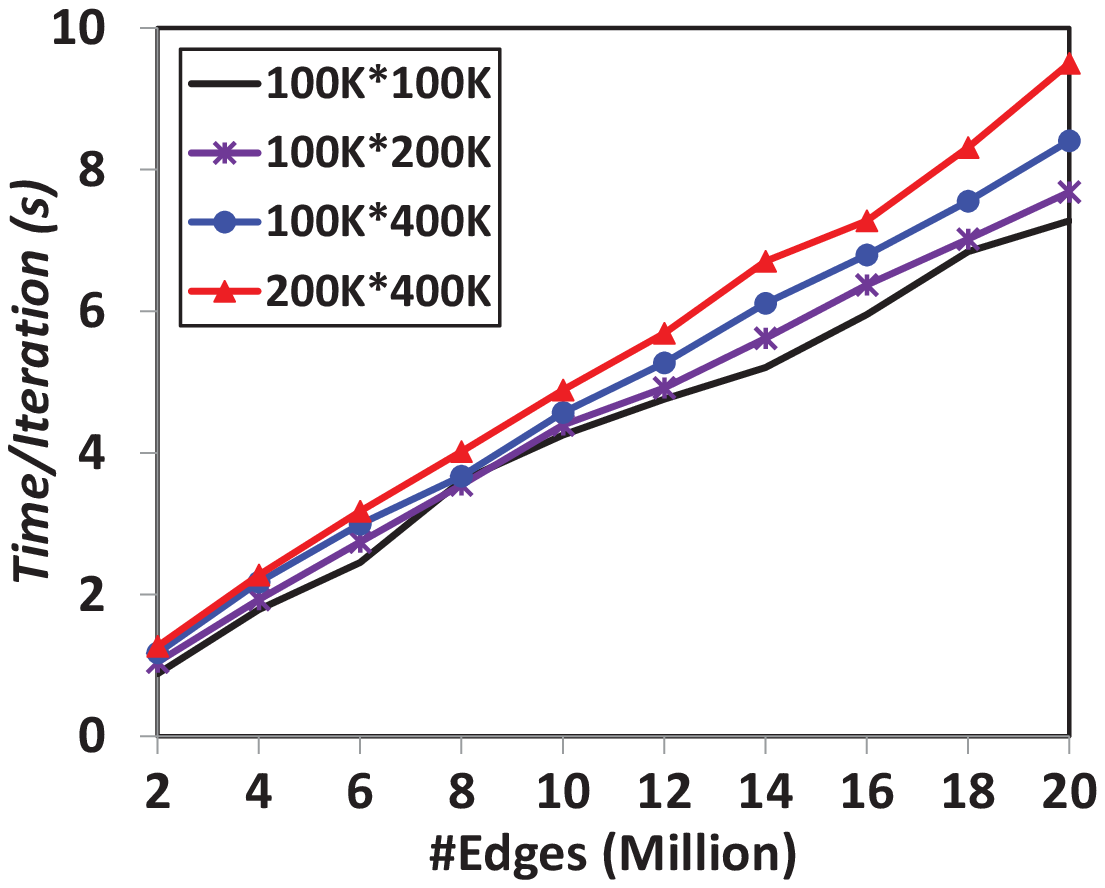}
		\caption{Random Graphs} \vspace{-8pt}
	\end{subfigure} \hspace{+10pt}
	\begin{subfigure}[b]{0.20\textwidth}
		\centering
		\includegraphics[width=\textwidth]{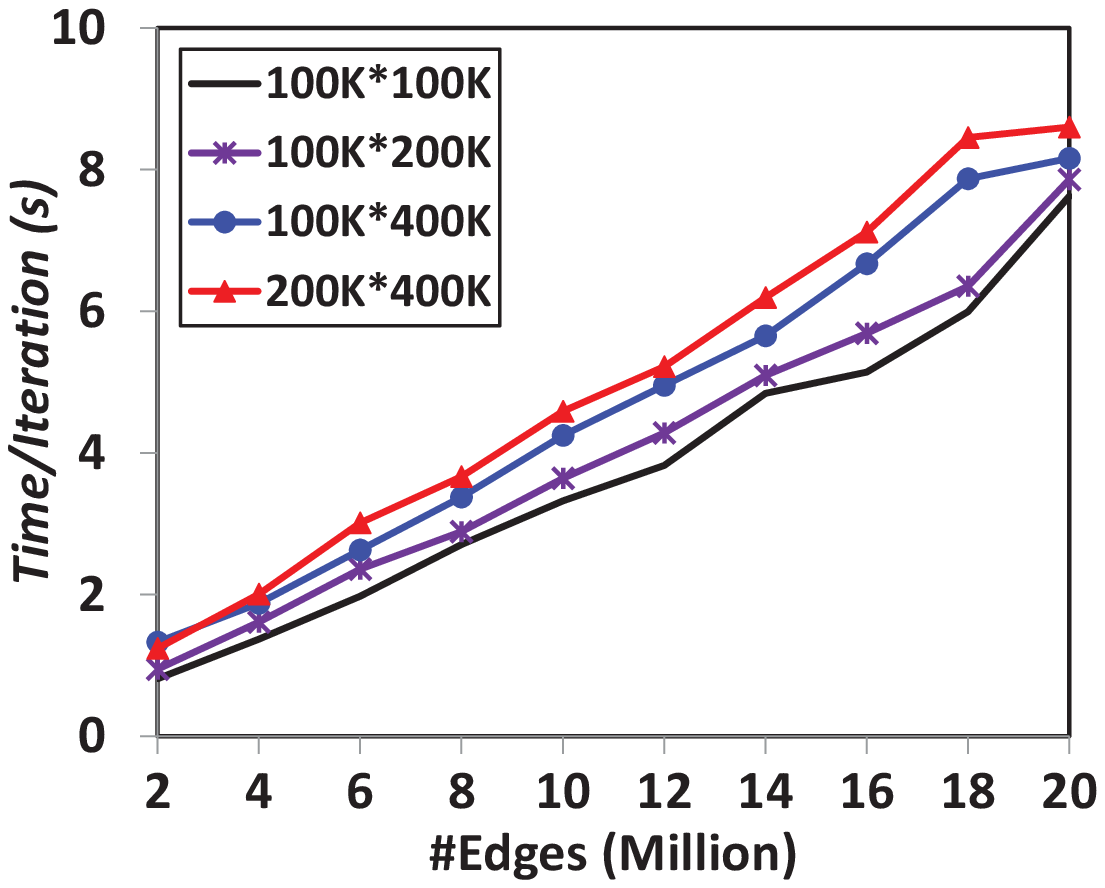}
		\caption{Power-law Graphs} \vspace{-8pt}
	\end{subfigure} 
	\caption{Running time per iteration \wrt number of edges.}
	\label{fig:time_efficiency}
	\vspace{-5pt}
\end{figure}

Figure~\ref{fig:time_efficiency} shows the average time per iteration for graphs of different settings. First, from each single line, we see that the actual running time per iteration exhibits linearity \wrt to number of edges in the graph.
More specifically, each iteration takes about 0.9 seconds for graphs with 2M edges, steadily increasing to $9$ seconds for graphs with 20M edges. This is rather efficient, given that we only run the algorithm in a single thread; for large-scale graphs, one can easily scale up the algorithm by parallelizing the matrix operations in multiple threads. 
Comparing across lines, we find that graphs of larger size take more time but still within the same magnitude.  As the edge count is the same, the additional time is due to traversing additional vertices in such larger graphs when performing matrix operations. 

\subsection{Evaluation of Popularity Prediction}
We now evaluate how does BiRank perform for the application of comment-based popularity prediction. 

\subsubsection{Experimental Settings}
\textbf{Datasets and metrics.}
Table \ref{tab:dataset} shows the demographics of the three real-world datasets used in this evaluation. 
Each dataset is constructed by the search results of some seed queries. More details about the dataset are given in \cite{sigir14:he}.
The evaluation ground-truth~(GT) is the number of views (note, not the number of comments) received in the future three days after the original crawl date~(\ie ranking date $t_0$).
\vspace{-5pt}
\begin{table}[H]
	\begin{center}
		\scriptsize
		\caption{Demographics of the three Web~2.0 datasets.}
		\vspace{-8pt}
		\label{tab:dataset}
		\begin{tabular}{ | l | c | c | c | c| c |}
			\hline
			\textbf{Dataset} & \textbf{Item\#} & \textbf{User\#} & \textbf{Comment\#} & \textbf{Avg C/I} & \textbf{Crawled Date} \\ \hline
			YouTube	& 21,653	& 3,620,487 	& 7,246,287 &	334.7 & 2012/8/9 \\ \hline
			Flickr	& 26,815	& 37,690  & 169,150	& 6.3 & 2012/9/3\\ \hline
			Last.fm	& 16,284	& 77,996 	& 530,237	& 32.6 & 2012/10/24\\ \hline
		\end{tabular}
	\end{center}
\end{table}
\vspace{-10pt}

Given a set of items, BiRank outputs a ranking list of the items, indicating their predicted popularity. To assess the quality of the predicted ranking with the GT ranking globally, we adopt the Spearman coefficient, which measures the agreement between two rankings.\\

\vspace{-5pt}
\noindent\textbf{Baselines.}
We compare with the following six baselines:

\noindent 1. \textbf{View Count~(VC)}: Rank based on the current view count of items, corresponding to our Hypothesis \textit{H3} alone.

\noindent 2. \textbf{Comment Count in the Past~(CCP)}: Rank based on the number of comments received in the 3-day period prior to $t_0$, corresponding to our Hypothesis~\textit{H1}.

\noindent 3. \textbf{Multivariate Linear model~(ML)~\cite{Pinto:2013}}: A state-of-the-art regression method for popularity prediction.
We apply this method on the comment series with the time unit as 3 days. This baseline is to test the traditional view-based methods when applied to modeling comments.

\noindent 4. \textbf{PageRank~\cite{pagerank:1999}}: This is the most widely used graph ranking method. Since the bipartite nature can cause the random walk to be non-stationary, 
we employ the standard method to set a uniform self-transition weight $w_{ii}=1$ for all nodes before converting to a stochastic matrix. 

\noindent 5. \textbf{Co-HITS~\cite{Co-HITS:2009}}: This algorithm is devised for ranking on bipartite graphs by interleaving two random walks. To make a fair algorithmic comparison with BiRank, we apply the same query vectors to Co-HITS and tune the parameters in the same way.

\noindent 6. \textbf{BGER~\cite{Cao2011}}: This is another algorithm designed for ranking on bipartite graphs. Instead of simulating a random walk, it normalizes the edge weights in a different way and is analogously explained as heat diffusion. We apply the same query vectors and parameter search for this method. 

To expedite parameter tuning, we randomly held out 10\% of the dataset as the development set, and employ grid search to find the optimal parameters. Then the performance is evaluated on the remaining 90\% as the testing items.

\subsubsection{Performance Comparison}

Table~\ref{tab:oe} shows the performance of the methods on the three datasets. First, we can see that the three bipartite graph ranking methods (lines 5 -- 7) significantly outperform other methods. This is because these methods model all the three ranking hypotheses we proposed, while other methods only partially model the hypotheses. Among the three bipartite ranking methods, BiRank achieves the best performance in general (best on two datasets YouTube and Flickr), followed by Co-HITS (best on Last.fm) and BGER. Further experimentation of 10-fold cross validation shows that the improvements of BiRank over Co-HITS and BGER on YouTube and Flickr datasets are consistent and statistically significant ($p < 0.01$, via one-sample paired t-test). Moreover, Co-HITS outperforms BGER consistently, although the random walk treatments of Co-HITS are suspicious to bias the high-degree vertices while BGER does not have this issue.
We suspect the reason of Co-HITS's strong performance might be that the bias effect is diluted by the setting of query vectors, which can regulate the random walks effectively. 

Focusing on the result of PageRank (line 4), we see that it performs very poorly for Flickr and Last.fm. This indicates that just the centrality of an item in the user--item temporal graph is insufficient for accurate popularity prediction. In addition, the performance discrepancy between PageRank and CoHITS (also a random walk-based method) highlights the importance of separately handling the two vertex types within the bipartite graph.

It is surprising that the regression approach ML underperforms CCP, as ML leverages more information: comments in the recent 30 days compared with CCP's access to only three days.
We believe there are two reasons for this: 1) the nature of short-term prediction, and 2) the sparsity of comments.
As the prediction task is a short-term one, the most recent data carries the most signal -- ``What happened yesterday will happen tomorrow''; the performance score of CCP verifies this point. Secondly, the sparsity in comment series (\eg some time units have zero count) can negatively affect the regression process in an unexpected manner.


\begin{table}
	\begin{center}
		\caption{Popularity prediction evaluated by Spearman coefficient (\%). ``*'' denotes the statistical significance for $p < 0.01$ judged by the one-sample paired t-test.}
		\vspace{-8pt}
		\label{tab:oe}
		\begin{tabular}{ | l | c | c | c |}
			\hline
			\textbf{Method} & \textbf{YouTube} & \textbf{Flickr} & \textbf{Last.fm} \\ \hline 
			1. VC	& 73.39 &	58.42	 &	67.31  \\ \hline
			2. CCP	& 83.35 &	59.43	 &	67.21  \\ \hline
			3. ML & 78.24 &	58.00  & 	38.09  \\ \hline
			4. PageRank & 80.72 &	28.15	& 	10.24  \\ \hline
			5. Co-HITS	& 85.21	& 63.81	&	\textbf{72.71}$^*$	\\ \hline
			6. BGER	& 84.10	& 63.17	&	68.94	\\ \hline
			7. BiRank (ours) &	\textbf{88.21}$^*$ &	\textbf{64.76}$^*$ &	70.93 \\ \hline
		\end{tabular}
	\end{center}
	\vspace{-15pt}
\end{table}


\subsection{Evaluation of Personalized Recommendation}
In this subsection, we study how do our BiRank and TriRank perform for the task of personalized item recommendation. 
\vspace{-5pt}

\subsubsection{Experimental Settings}
\textbf{Datasets.} We experiment with two public review datasets: Yelp\footnote{\url{yelp.com/dataset_challenge}. Downloaded on October 2014.} and Amazon Electronics\footnote{\url{snap.stanford.edu/data/web-Amazon-links.html}}. We follow the common practice in evaluating recommendation algorithms~\cite{BPR:2009, He:SIGIR16} that filters out users and items with fewer than 10 reviews. 
We used the sentiment analysis tool developed by \cite{Sentiment:2014} for extracting aspects from review texts.
Table~\ref{tab:recom_dataset} summarizes the statistics of the filtered datasets and Table~\ref{tab:aspect_examples} shows examples of the top aspects extracted.
\vspace{-8pt}
\begin{table}[h]
	\begin{center}
		\caption{\textbf{Statistics of datasets in evaluation.}}
		\vspace{-8pt}
		\small
		\label{tab:recom_dataset}
		\begin{tabular}{ | l | c | c | c | c | c | }
			\hline
			\textbf{Dataset} & \textbf{Review\#} & \textbf{Item\#} & \textbf{User\#} & \textbf{Aspect\#} \\ \hline
			Yelp	& 114,316	& 4,043 & 3,835 & 6,025 \\ \hline
			Amazon	& 55,677	& 14,370& 2,933 & 1,617 \\ \hline
		\end{tabular}
	\end{center}
	\vspace{-7pt}
\end{table}

\noindent\textbf{Baselines.} We compare with the following methods that are commonly used in top-$K$ recommendation:

\noindent 1. \textbf{Popularity~(ItemPop)}. Items are ranked by their popularity judged by number of ratings. This non-personalized method benchmarks the performance of the top-$K$ task. 

\noindent 2. \textbf{ItemKNN~\cite{ItemKNN:WWW2001}}. This is standard item-based CF. We tested the method with different number of neighbors, finding that using all neighbors works best.

\noindent 3. \textbf{PureSVD~\cite{Cremonesi:Recsys2010}}. A state-of-the-art model-based CF method for top-$K$ recommendation, which performs SVD on the whole matrix. We tuned the number of latent factors.

\noindent 4. \textbf{PageRank~\cite{Haveliwala:2002}}.
This graph method has been widely used for top-$K$ recommendation, such as by \cite{Lee:Recsys2011}. For a fair comparison, we set the personalized vector the same with TriRank's query vectors and tuned the damping factor.

\noindent 5. \textbf{TagRW~\cite{TagRW:TMIS2013}}. A state-of-the-art tag-based recommendation solution, which performs random walks on the user--user and item--item similarity graph. Since tags have a similar form with aspects, we feed aspect as tags into the method. 

For each user, we sort her reviews in chronological order. The first 80\% are used for training, followed by 10\% as validation (for parameter tuning) and 10\% as test set (for evaluation).
Given a test user, we assess the ranked list of top-$K$ items with \textit{Hit Ratio}~\cite{BPR:2009} and \textit{NDCG}~\cite{He:2015}.

\subsubsection{Performance Comparison}

\begin{figure}[t]
	\centering
	\begin{subfigure}[b]{0.24\textwidth}
		\centering
		\includegraphics[width=\textwidth]{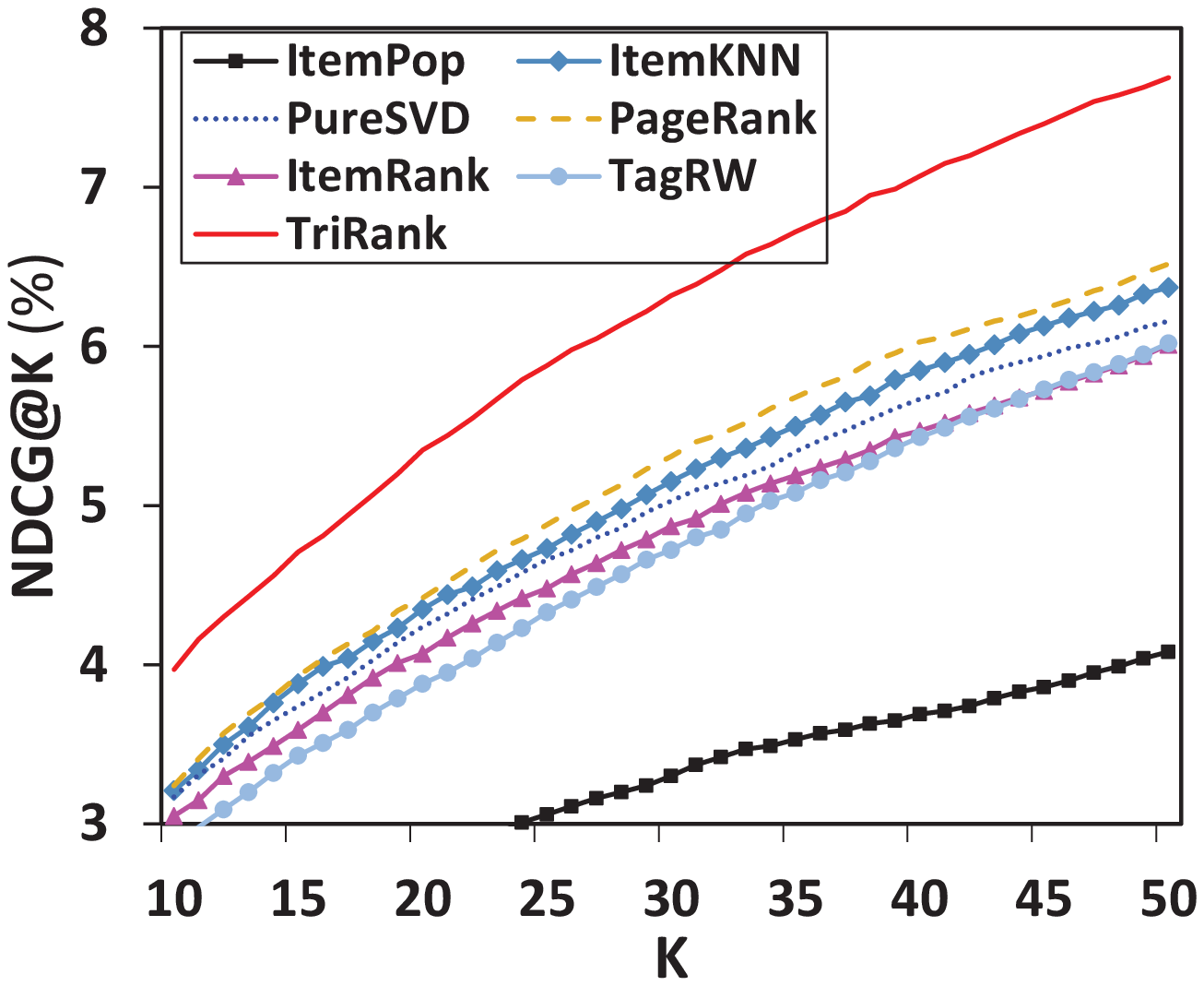}
		\vspace{-15pt}
		\caption{Yelp -- NDCG}
		\label{fig:yelp_ndcg}
	\end{subfigure} \hspace{-7pt}
	\begin{subfigure}[b]{0.24\textwidth}
		\centering
		\includegraphics[width=\textwidth]{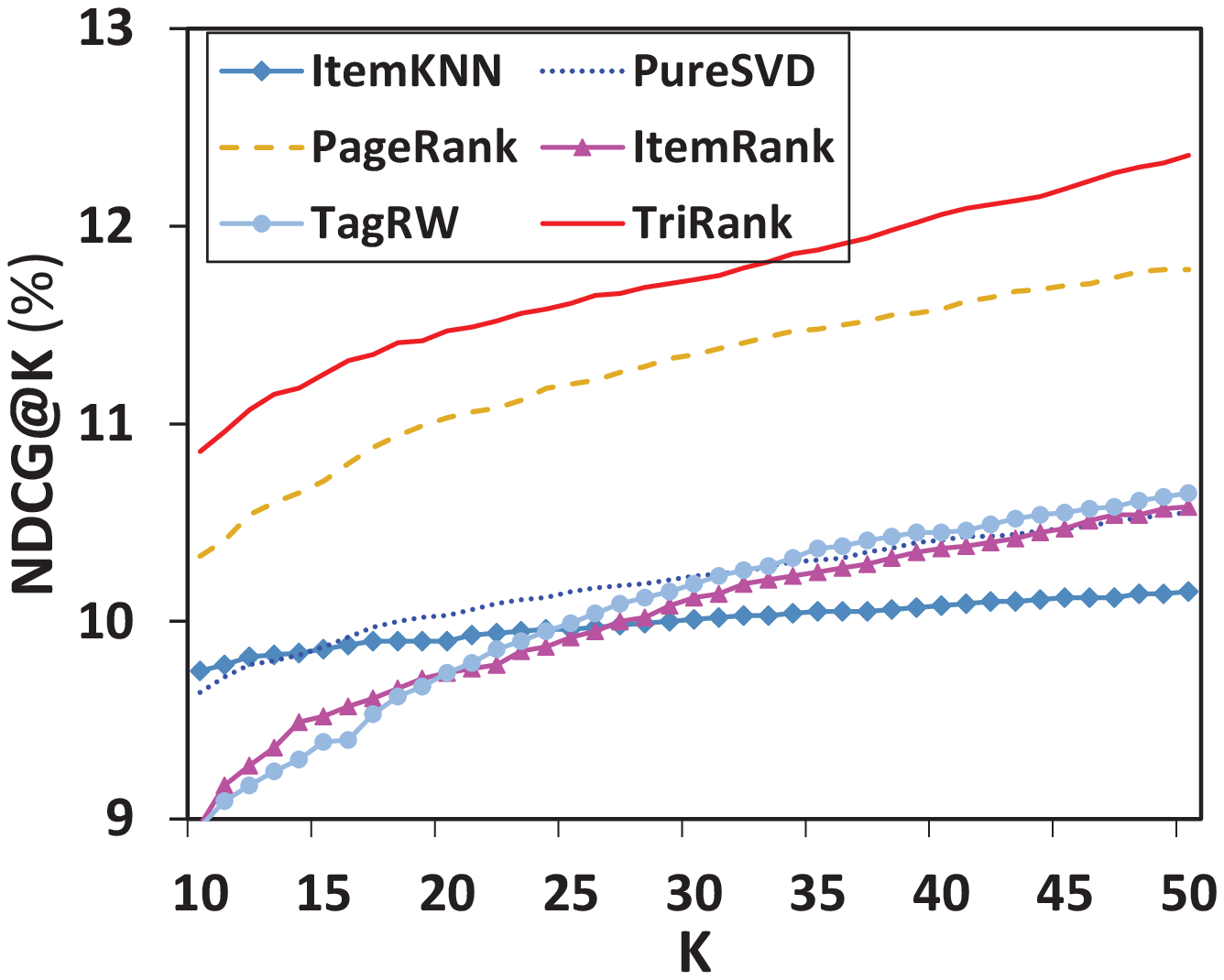}
		\vspace{-15pt}
		\caption{Amazon -- NDCG}
		\label{fig:elec_ndcg}
	\end{subfigure} \hspace{-7pt}
	\vspace{-5pt}
	\caption{\textbf{Performance comparison of top-$K$ recommendation evaluated by NDCG from position 10 to 50~(\ie $K$).}}
	\vspace{-10pt}
	\label{fig:recom_performance}
\end{figure}

Figure~\ref{fig:recom_performance} plots the performance of top-$K$ recommendation methods evaluated by NDCG from position 10 to 50. Performance of hit ratio shows a similar trend with NDCG, and is thus omitted for space. ItemPop performed very weakly on the Amazon dataset, and is entirely omitted in Figure~\ref{fig:elec_ndcg} to better highlight the performance of the other methods. As can be seen, our TriRank consistently outperforms the baselines with a large margin, and the one-sample paired t-test verifies that the improvements over all baselines are statistically significant with $p < 0.01$. For a more detailed discussion, we further show the concrete scores obtained at the position~50 in Table~\ref{tab:recom_performance}.


Focusing on Lines~1--4 that are all CF methods that only model the user--item relationship, we see that our BiRank achieves the best performance on both datasets; specifically, it improves over the competitive recommendation methods ItemKNN and PureSVD with a relative improvement about $8.3\%$. This is very encouraging, and gives evidence of the merit of our specification of BiRank (in Section~{\ref{ss:recom_birank}}) for collaborative filtering.
ItemKNN performs very well on the Yelp dataset (better than PureSVD), but poorly on the Amazon one. One possible reason comes from data sparsity: as in Table~\ref{tab:recom_dataset}, each item of the Amazon dataset only has 3.9 reviews on average. In such cases, the statistical similarity measure may fail in neighbor-based CF.  In contrast, model-based methods are more robust to sparse data by projecting users and items to the latent space. Lastly, we see Item Popularity performs the worst, indicating the importance of modeling users’ personalized preferences, rather than just recommending popular items.

Moving to Lines~5--7 of review-based methods, we see that they generally improve over the methods that use CF only, indicating the utility of reviews (more specifically, item aspects) for uncovering users' preference and complementing with user ratings. Second, TriRank achieves the best performance, further improving over BiRank with over a 10\% relative improvement and outperforming PageRank and TagRW significantly. 
This verifies the effectiveness of our TriRank in incorporating the aspects for enhanced recommendation. 
Lastly, TagRW is inferior to PageRank in utilizing the same aspect source. We believe the main reason comes from TagRW's transformation of the user--item--aspect graph to user--user and item--item graphs, which can cause some signal loss especially when the original relationships are sparse. 

\begin{table}[t]
	\begin{center}
		\caption{Recommendation performance ($\%$) evaluated at Rank~50. BiRank outperforms CF-based methods (Lines 1--3) and TriRank outperforms all other methods.} 
		\label{tab:recom_performance}
		\vspace{-8pt}
		\begin{tabular}{|l|c|c|c|c|}
			\hline
			\textbf{Dataset} & \multicolumn{2}{|c|}{\textbf{Yelp}} & \multicolumn{2}{|c|}{\textbf{Amazon}} \\ \hline 
			\textbf{Metric}(\%) & \textbf{HR} & \textbf{NDCG} & \textbf{HR} & \textbf{NDCG} \\ \hline
			1. ItemPop		& $10.61$ & $4.08$ 	& $6.13$	& $2.37$  \\ \hline
			2. ItemKNN 	& $15.72$ & $6.37$	& $12.69$	& $10.15$ \\ \hline
			3. PureSVD	& $14.94$ & $6.16$  & $14.94$	& $10.55$ \\ \hline
			4. BiRank (ours) & $\textbf{17.00}^*$ & $\textbf{6.90}^*$	& $\textbf{15.97}^*$	& $\textbf{11.16}^*$ \\ \hline \hline
			5. PageRank & $15.90$ & $6.52$  & $17.49$	& $11.78$ \\ \hline
			6. TagRW    & $15.25$ & $6.02$  & $17.47$	& $10.65$ \\ \hline
			7. TriRank (ours) & $\textbf{18.58}^*$ & $\textbf{7.69}^*$ & $\textbf{18.44}^*$	&$\textbf{12.36}^*$ \\ \hline
		\end{tabular}
		\vspace{-10pt}
	\end{center}
\end{table}

\section{Conclusion}
\label{sec:conclusion}
We focus on the problem of ranking vertices of bipartite graphs, and more generally, $n$-partite graphs. 
We devise a new, generic algorithm -- BiRank -- which ranks vertices by accounting for both the graph structure and prior knowledge. 
BiRank is theoretically guaranteed to converge to a stationary solution, and can be explained from both a regularization view and a Bayesian view. This appealing feature allows future extensions to BiRank to be grounded in a principled way. 
To demonstrate the efficacy of our proposal, we examine two ranking scenarios: a general ranking scenario of item popularity prediction
by modeling the user--item binary relationship, 
and a personalized ranking scenario of item recommendation by modeling the user--item--aspect ternary relationship.
By properly setting the graph's edge weights and query vectors, BiRank can be customized to encode various ranking hypotheses.
Extensive experiments on both synthetic and real datasets 
demonstrate the effectiveness of our method.
In future, we will study how to optimally learn the hyper-parameters of BiRank. Owing to the two views of BiRank, two solutions can be explored --- by adapting the parameters based on the validation set~\cite{Rendle:2012},
or by integrating over the parameters under the Bayesian network formalism.
Moreover, we will explore how to integrate the graph regularization framework with matrix factorization methods, which have been shown to be very effective for many tasks such as recommendation~\cite{He:SIGIR16} and clustering~\cite{He:WWW2014}. 


\vspace{-10pt}
\section*{Acknowledgments}
{This research is supported by the National Key Research and Development Program of China (No. 2016YFB1000905) and NSFC under Grant No. U1401256. NExT research is supported by the National Research Foundation, Prime Minister's Office, Singapore under its IRC@SG Funding Initiative. The authors thank the anonymous reviewers for their valuable comments, and acknowledge the additional discussion and help from Liqiang Nie, Jun-Ping Ng, Tao Chen and Yiqun Liu. This paper is an extended version of the SIGIR '14 conference paper~\cite{sigir14:he}. Xiangnan He is the corresponding author. 
}

\bibliographystyle{IEEEtran}

\vspace{-30pt}

\begin{IEEEbiography}[{\includegraphics[width=1in,height=1.25in,clip,keepaspectratio]{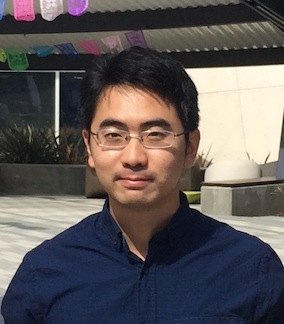}}]{Xiangnan He}
	is currently a postdoctoral research fellow with the School of Computing at the National University of Singapore. His research interests include information retrieval, recommender systems, multimedia and machine learning. His works have appeared in several major international conferences such as SIGIR, WWW, MM, CIKM and AAAI. He has served as program committee member of international conferences such as SIGIR, WWW and EMNLP. 
\end{IEEEbiography} \vspace{-35pt}
\begin{IEEEbiography}[{\includegraphics[width=1in,height=1.25in,clip,keepaspectratio]{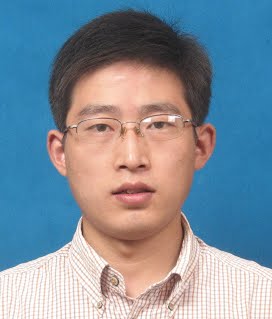}}]{Ming Gao}
	is an associate professor of Software Engineering Institute with the East China Normal University, China. He received his doctorate from the School of Computer Science, Fudan University. 
	His research interests include uncertain data management, streaming data processing, social network analysis and data mining. His work appears in major international conferences including ICDE, ICDM, DASFAA and WebSci.
\end{IEEEbiography} \vspace{-35pt}
\begin{IEEEbiography}[{\includegraphics[width=1in,height=1.25in,clip,keepaspectratio]{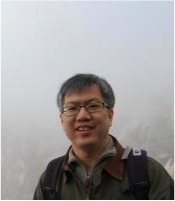}}]{Min-Yen Kan}
	is an associate professor at the National University of Singapore.  
	He previously served as a member of the executive committee of the Association of Computational Linguistics (ACL) and maintains the ACL Anthology, the community's largest archive of published research.  He is an associate editor for the Springer ``Information Retrieval'' journal. His research interests include digital libraries and applied natural language processing and information retrieval.  
	Specific projects include work in the areas of scientific discourse analysis, full-text literature mining, machine translation, lexical semantics and applied text summarization.
\end{IEEEbiography} \vspace{-20pt}
\begin{IEEEbiography}[{\includegraphics[width=1in,height=1.25in,clip,keepaspectratio]{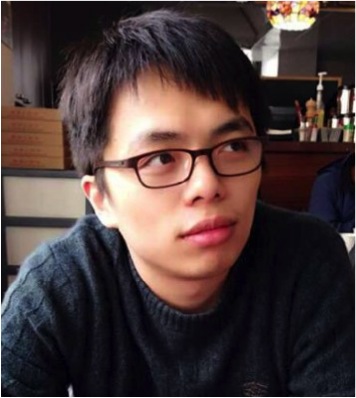}}]{Dingxian Wang}
	is a Research Engineer with the Search Science Department in eBay. He is specifically working on the project of product classification for improving search ranking. His research interests mainly include information retrieval, software engineering, natural language processing and applied machine learning. His works have appeared in international conferences including WISE, CSE and ICSSP. 
\end{IEEEbiography}
\vfill
\end{document}